\documentclass[acmsmall]{acmart}
\setcopyright{cc}
\setcctype{by}
\acmDOI{10.1145/3808274}
\acmYear{2026}
\acmJournal{PACMPL}
\acmVolume{10}
\acmNumber{PLDI}
\acmArticle{196}
\acmMonth{6}
\acmSubmissionID{pldi26main-p135-p}
\received{2025-11-13}
\received[accepted]{2026-04-03}

\usepackage{listings}

\usepackage{graphicx}
\usepackage{framed}
\usepackage{amsmath,amsfonts,amsthm, mathrsfs}
\usepackage{graphicx}
\usepackage{enumerate}
\usepackage{xcolor, varwidth}
\usepackage{amsmath, amsthm, wasysym, verbatim, bbm, color, graphics, geometry, multicol, tikz, wrapfig, harpoon, array, graphicx}
\usepackage{pgfplots, subcaption}
\usepackage{algpseudocode, algorithm}
\usetikzlibrary{shapes, patterns, calc}
\usetikzlibrary{automata, positioning, arrows}
\usepackage{tabularx}
\usepackage{wrapfig}

\pgfplotsset{compat=1.18} 

\usepackage{thmtools, thm-restate}
\usepackage{enumitem}
\usepackage[frozencache,cachedir=.]{minted}

\theoremstyle{plain}
\newtheorem{theorem}{Theorem}[section]
\newtheorem{corollary}[theorem]{Corollary}
\newtheorem{lemma}[theorem]{Lemma}

\theoremstyle{definition}
\newtheorem{definition}[theorem]{Definition}

\newtheorem{example}[theorem]{Example}

\numberwithin{equation}{section}
\numberwithin{figure}{section}
\numberwithin{table}{section}
\numberwithin{algorithm}{section}

\newcommand\GG{\mathcal{G}}
\newcommand\VV{\mathcal{V}}
\newcommand\EE{\mathcal{E}}
\newcommand\nf{\texttt{nodefun}}
\newcommand\gf{\texttt{graphfun}}

\newcommand\fpath{f^\textsf{path}}
\newcommand\pa{\textsf{Pa}}
\newcommand\anc[1]{\textsf{Anc}_Z^*(#1)}
\newcommand\dir{\textsf{dir}}

\begin{document}

\title{Causality and Semantic Separation}

\author{Anna Zhang}
\orcid{0009-0000-6247-5667}
\affiliation{%
  \institution{Massachusetts Institute of Technology}
  \city{Cambridge}
  \country{USA}
}
\email{annazhang@alum.mit.edu}

\author{Qinglan Luo}
\orcid{0009-0007-7913-0092}
\affiliation{%
  \institution{Wellesley College}
  \city{Wellesley}
  \country{USA}
}
\affiliation{%
  \institution{Massachusetts Institute of Technology}
  \city{Cambridge}
  \country{USA}
}
\email{ql101@mit.edu}

\author{London Bielicke}
\orcid{0009-0002-3865-4707}
\affiliation{%
  \institution{University of California at Los Angeles}
  \city{Los Angeles}
  \country{USA}
}
\email{londonbielicke@g.ucla.edu}

\author{Eunice Jun}
\orcid{0000-0002-4050-4284}
\affiliation{%
  \institution{University of California at Los Angeles}
  \city{Los Angeles}
  \country{USA}
}
\email{emjun@ucla.edu}

\author{Adam Chlipala}
\orcid{0000-0001-7085-9417}
\affiliation{%
  \institution{Massachusetts Institute of Technology}
  \city{Cambridge}
  \country{USA}
}
\email{adamc@csail.mit.edu}



\begin{abstract}
  The design of scientific experiments deserves its own variation of formal
  verification to catch cases where scientists made important mistakes, such as
  forgetting to take confounding variables into account. One of the most
  fundamental underpinnings of science is \emph{causality}, or what it means for
  interventions in the world to \emph{cause} other outcomes, as formalized by
  computer scientists like Judea Pearl. However, these ideas had not previously
  been made rigorous to the standards of the programming-languages community,
  where one expects a (syntactic) program analysis to be proved sound with
  respect to a natural semantics. In the domain of causality, as the relevant
  ``program analysis,'' we focus on $d$-separation, a classic condition on graphs that can be used to decide when the
  design of an experiment controls for sufficiently many confounding variables, even though the reason that this condition works is often unintuitive. Our central
  result (mechanized in Rocq) is that $d$-separation exactly coincides with a
  novel \emph{semantic} definition inspired by
  noninterference from the theory of security. This characterization provides a structural semantic foundation for $d$-separation and helps explain why the graph-theoretic condition is correct, independently of probabilistic assumptions. 
  For each given automated test on
  the quality of an experiment design, our theorem justifies an associated
  method for falsifying the world-modeling hypothesis behind the experiment.
\end{abstract}

\maketitle

\section{Introduction}

A scientific experiment can be seen as a program, perhaps a complex one, that
performs particular interactions with its environment, with the goal of refining 
an understanding of how that environment works. That is, the scientists begin
with a \emph{world model} that identifies a set (often infinite) of possible
programs that could represent the environment. Through the experiment, a scientist
manages to shrink the set of world models that remain plausible. The hazards in
experimental design are so common and diverse that the process of enumerating them is known as assessing \emph{threats to
validity}.
How is a
scientist to know an experimental design accounts for all of the important
hazards? We found that question similar to the question of how a programmer
knows that a program behaves properly, which leads to the whole discipline of
program verification. Noting the apparent lack of such a discipline for
experiment design, we set out to start building up its foundations, which led us
to the specific element of \emph{causality} that this paper covers.

First, let us review the elements of program verification often taken for
granted, to help in drawing connections to experiment design. Nearly any program
verification involves \emph{assumptions} and \emph{proof obligations}. The
assumptions could include the semantics of the language that the program is written
in or the behavior of external libraries to which the program links. Then we
have a \emph{specification} covering how the program is allowed to behave, when
the assumptions hold.

Similar elements are apparent in experiment design. An experiment runs in some
real-world setting, where it is reasonable to model the world as a program. We
do not know exactly which program, but we aim to learn more about the
possibilities. If we simply fix the ``type signature'' of the world in terms of
which measurable quantities exist, then we allow too broad a space of programs, where
it is hopeless to use observation to narrow the possibilities significantly.
Thus, just as a program verification may assume a specification for library
functions, an experiment assumes a \emph{world model} that fixes a class of
programs that ought to include the \emph{true} world program.
Then the experiment, the program to be verified, interacts with the world and uses what it learns to rule out possible world models. 
So \emph{the soundness of logic
for ruling out possible worlds is central to science}.

To make this idea more tangible, consider a simple example from medicine. Suppose we are testing a new treatment for a heart condition, and a scientist’s model assumes that gender has no direct causal effect on recovery rates once treatment and age are controlled for. This assumption corresponds to a particular structure in the model’s causal graph. If, in a randomized clinical trial, we observe that outcomes still vary systematically by gender after conditioning on treatment and age, then the experiment has falsified that model. In this sense, the experiment serves as a test oracle for the world model: it probes the model’s assumptions through carefully designed manipulations.

Our framework aims to formalize the connection between the structure of a causal model and the outcomes of well-controlled experiments in the same way that programming-language semantics formalize what can be proved about a program’s behavior. Program verification often relies on tractable automated analyses of programs,
which should be proved sound with respect to semantics. For instance, an
abstract interpretation~\cite{AbsInt} is used to assign overapproximating 
syntactic descriptions to the variables of a program, iterating to find
a fixed point in a lattice. We define soundness with respect to the assumed
semantics of the programming language. In experiment design, there are similar
techniques that generations of scientists have relied upon. The field of
\emph{causal modeling} 
defines the analogue of programs in a
language: a \emph{structural causal model} or \emph{causal model}, a DAG whose nodes are fundamental quantities
in the world (perhaps observable, perhaps not), where directed edges indicate
(potential) causal relationships between nodes. There is also the equivalent
of a set of program analyses.
In this paper, we focus on \emph{$\mathrm{d}$-separation}, a decidable property that characterizes when a set of nodes in a DAG blocks all influences between two others -- for example, when conditioning on that set avoids spurious conclusions due to confounding.
Although the definition of $d$-separation is simple, why it works is not obvious. We therefore set out to prove a relationship with a semantic characterization that explains the correctness of the graph-theoretic condition.

We were inspired by two families of software-related techniques.
\begin{enumerate}
\item \emph{Information-flow security} drives reasoning about how programs
protect secrets and block undue influence by their environments through
controlling how information flows between different inputs, outputs, and state
elements.  The gold-standard semantic property is
\emph{noninterference}~\cite{noninterference}, for instance characterizing
confidentiality by saying that varying secret inputs to a program must never
change public outputs.  Similar kinds of information-flow restrictions are
central to the value of causal DAGs in restricting the sets of programs needed
to model real-world phenomena, but the field of causality had lacked a similarly
elegant \emph{semantic} explanation of what $d$-separation checks guarantee.
\item \emph{Program testing} allows us to find bugs by subjecting programs to
clever sequences of inputs, designed to exercise corner cases.  Regardless of
the clever test suite, we need a \emph{test oracle} that evaluates program
behavior against a specification.  As an application of our new semantic
characterization of $d$-separation, we demonstrate how to prove soundness of test
oracles, in the sense that, through observation of the world, they are able to
\emph{falsify} world models.  In a (somewhat subtle) sense, we also prove that
certain oracles are \emph{complete}, meaning if we happened to know that no
possible test could fail the oracle, then the $d$-separation syntactic check must
succeed on the DAG.
\end{enumerate}

So, summing up, we introduce a new semantic condition
characterizing what $d$-separation establishes, and we start to explore its
consequences for justifying the sound design of experiments. Our central result
is fully mechanized in the Rocq Prover, and we provide the code as an associated artifact \cite{zhang_2026_19075958}. 

We now review the relevant concepts from the literature on experimental design and causal reasoning.
Then, we present our new semantic characterization of $d$-separation and our main
theorem justifying it (soundness and completeness). Afterward, we examine
implications of that theorem: a more intuitive way of understanding what
$d$-separation tells us, which allows us to prove soundness and completeness for
test oracles to use in experiments. With the payoff established, we conclude by sketching the proof
of our central theorem.

\section{Background}

\subsection{Experimental Design and Validity}
The purpose of scientific experiments is to establish causal relationships. 
For example, randomized controlled trials are the gold standard for establishing causality in medicine.
Experiments are well-designed if they provide unbiased, valid causal estimates. 
In practice, ideal designs are expensive, requiring many participants and resources. 
Certain design choices are more efficient but may inadvertently introduce confounding.
Confounding, or the presence of variables that influence both causes and effects, can lead to an inability to calculate causal estimands~\cite{cinelli2020controls}. 


In \textit{Experimental and Quasi-experimental Designs for Generalized Causal Inference}, \citet{shadish2002validity} put forth a now-prevalent theory of scientific validity for experiments. 
They outline four types of validity: external, construct, internal, and
statistical-conclusion. 
External validity refers to the generalizability of experimental findings.
Construct validity pertains to how well the variables measured (e.g., IQ test
results vs. SAT scores) represent the intended construct (e.g., intelligence).
Internal validity in an experiment refers to the accurate estimation of causal conclusions.
Statistical-conclusion validity relates to the appropriate and correct
formulation of statistical models to analyze data collected from an experiment.
While Shadish et al.\ provide these definitions and enumerate
threats to validity, or ways in which each type of validity can be compromised,
no work has since provided a formal, principled way to detect
scientific validity automatically in experiments. The closest set of tools for diagnosing
validity come in the form of graphical criteria (e.g., backdoor criterion) using Pearl's
causal diagrams~\cite{pearl2009causality}, as discussed below. Our work is a step towards
connecting experimental design with causal reasoning by establishing a formal
framework for analyzing the semantics of causal diagrams. 

Synthesizing the data collected through a scientific experiment to establish
evidence for causal relationships involves the following steps. First,
researchers must formulate a causal model that encodes their hypotheses
about the structure of causal relationships
as a DAG. This model must be falsifiable. It should make
predictions that can be tested and potentially contradicted by observed data.
Second, researchers assign probabilities to values of nodes (variables) and use
statistical methods (e.g., linear regression, Bayesian inference) to update
these probabilities in light of experimental data.
Our work is primarily focused on the first step. We formalize how experimental
designs support or challenge the falsifiability of a given causal model. 


\subsection{Frameworks of Causality}
Structural causal models (SCMs)~\cite{pearl2009causality} represent causal
relationships using DAGs, where nodes correspond to
variables, and edges represent direct causal effects. SCMs encode assumptions
about causal structure explicitly and support reasoning about interventions
using the do-operator, which simulates external manipulation of variables (i.e., random assignment in an experiment).

The Potential Outcomes framework (PO)~\cite{rubin1974estimating, imbens2015causal} characterizes
causal effects by comparing counterfactual outcomes for the same
unit (e.g., patient) under different treatments (e.g., drug vs. placebo).
While PO provides an instance-based view,
SCMs provide a graphical, compositional view of causal networks
from which to derive counterfactuals. 
We focus on SCMs due to their structural similarity to information-flow security models (see \autoref{sec:relatedWork}).

\subsection{Fundamental Concepts in SCMs}
Understanding causality using SCMs begins with acknowledging that different types of
questions about the world require fundamentally different kinds of reasoning.
The hierarchy known as the ``Ladder of Causation'' \cite{bookofwhy} is helpful
for categorizing these levels of causal understanding:
association, intervention, and counterfactuals.
While correlations in observational data alone can answer questions at the first
level, answering intervention or counterfactual questions typically requires
assumptions about the underlying causal structure of the world. 
SCMs, represented as DAGs, support such reasoning. 




\begin{figure}
\begin{framed}
    \begin{subfigure}[b]{\textwidth}
    \centering
    \begin{tikzpicture}[
        node distance=1cm and 1cm,
        every node/.style={minimum size=1.2em, font=\small},
        circ/.style={draw, thick, circle, inner sep=2pt},
        dots/.style={font=\scriptsize},
        roundnode/.style={},
        unobserved/.style={font=\small, circle}
        ]

    \node[roundnode] (F) {\verb|days until finals week|};
    \node[roundnode, right=of F] (M) {\verb|students on campus|};
    \node[roundnode, right=of M] (B) {\verb|boba sales|};

    \draw[->, thick] (F) -- (M);
    \draw[->, thick] (M) -- (B);

    \end{tikzpicture}
    \caption{The number of students studying on campus \textbf{mediates} the effect of finals season on boba sales.}
    \label{fig:boba_sales}
    \end{subfigure}
    
    \begin{subfigure}[b]{\textwidth}
    \centering
    \vspace{.5em}
    \begin{tikzpicture}[
        node distance=1cm and 1cm,
        every node/.style={minimum size=1.2em, font=\small},
        circ/.style={draw, thick, circle, inner sep=2pt},
        dots/.style={font=\scriptsize},
        roundnode/.style={},
        unobserved/.style={font=\small, circle}
        ]

    \node[roundnode] (L) {\verb|courseload|};
    \node[roundnode, left=of L] (C) {\verb|caffeine|};
    \node[roundnode, right=of L] (G) {\verb|GPA|};

    \draw[->, thick] (L) -- (C);
    \draw[->, thick] (L) -- (G);

    \end{tikzpicture}
    \caption{From analysis of experimental data alone, it may appear that higher caffeine consumption is a cause of a lower GPA. However, the causal model shows that the \textbf{confounding} variable, a heavy courseload, is actually the cause of both.}

    \label{fig:caffeine}
    \end{subfigure}
    \begin{subfigure}[b]{\textwidth}
    \centering
    \begin{tikzpicture}[
        node distance=1cm and 1cm,
        every node/.style={minimum size=1.2em, font=\small},
        circ/.style={draw, thick, circle, inner sep=2pt},
        dots/.style={font=\scriptsize},
        roundnode/.style={},
        unobserved/.style={font=\small, circle}
        ]

    \node[roundnode] (D) {\verb|exam grades|};
    \node[roundnode, left=of D] (L) {\verb|procrastination|};
    \node[roundnode, right=of D] (T) {\verb|test anxiety|};

    \draw[->, thick] (L) -- (D);
    \draw[->, thick] (T) -- (D);

    \end{tikzpicture}
    \caption{While lots of procrastination and high test anxiety are otherwise unrelated, conditioning on the \textbf{collider} (e.g., looking only at students with low exam grades) can create a false association between the two.}
    \label{fig:exam_grades}
    \end{subfigure}
    
    \caption{These causal models illustrate examples of mediators, confounders, and colliders in academic settings.}
    \label{fig:examples}
\end{framed}
\end{figure}

It is useful to consider causal \textit{paths} because they represent possible
routes through which influence can flow among variables in a causal model.
These paths are undirected in the sense that they may contain edges pointing
both forward and backward, allowing us to examine dependencies that emerge from
various structural configurations. (This need to account for edges flowing in
both directions can perhaps be seen as the reason why we need to develop new
theoretical machinery, beyond what arises for noninterference in
information-flow security.) The following three node structures often appear in
such paths:
\begin{definition} \label{def:med_con_col}
    Let $\GG = (\VV, \EE)$ be a causal model, where $a, b, c \in \VV$.
    \begin{itemize}
        \item $b$ is a \textbf{mediator} of $a$ and $c$ if $(a,b) \in \EE$ and
        $(b,c) \in \EE$ or if $(b,a) \in \EE$ and $(c,b) \in \EE$.
        \item $b$ is a \textbf{confounder} of $a$ and $c$ if $(b,a) \in \EE$ and
        $(b, c) \in \EE$.
        \item $b$ is a \textbf{collider} of $a$ and $c$ if $(a,b) \in \EE$ and
        $(c,b) \in \EE$.
    \end{itemize}
\end{definition}
\begin{definition} \label{def:med_con_col_path}
    Let $\GG = (\VV, \EE)$ be a causal model and let $P$ be a path in $\GG$. If $a, b,$ and $c$ are three consecutive nodes on $P$, then $b$ is a \textbf{mediator, confounder, or collider on \boldmath$P$} according to which relationship from \autoref{def:med_con_col} holds among $a$, $b$, and $c$.
\end{definition}

Examples appear in \autoref{fig:examples}.
Understanding how to \textit{adjust}, or statistically account, for these
structures is central to causal inference. 
For instance, adjusting for a confounder helps remove estimation bias by ``blocking,'' or
mitigating, spurious associations, 
but adjusting for a collider can actually introduce bias by opening a non-causal path between variables.

More formally, causal relationships are naturally related to the concept of independence. 
In particular, a relevant property in a causal model is whether
two nodes $a$ and $b$ are independent conditioned on some subset of nodes $Z
\subseteq \VV$. If so, then an experiment conditioning on (e.g., through
randomization or sampling) the variables in $Z$ will ensure that $a$ has no
effect on $b$, and vice versa.  
This concept can be formalized further using probabilities.
\begin{definition}\label{def:ci_prob}
    Let $\VV$ be a finite set of variables, and let $\mathbf{P}(\cdot)$ be a probability distribution over $\VV$. Let $a,b \in \VV$ and $Z \subseteq \VV$. Then, variables $a$ and $b$ are \textbf{conditionally independent given \boldmath$Z$} if $$\mathbf{P}(a = \alpha \,|\, b = \beta, Z = (\zeta_1, ..., \zeta_k)) = \mathbf{P}(a = \alpha \,|\, Z = (\zeta_1, ..., \zeta_k))$$ for all possible assignments of values $\alpha, \beta, (\zeta_1, ..., \zeta_k)$.
\end{definition}
While this definition is precise, there is no clear mapping to reasoning about
causal graphs.\ The probabilistic notion of conditional independence is
inherently algebraic:\ it tells us that knowing $Z$ renders $a$ and $b$
independent, but it offers no guidance on how to determine this relationship
from the structure of a graph. In practice, repeatedly checking conditional
independencies from joint distributions is computationally expensive and
provides little insight into \textit{why} the independence holds. What we need
is a way (a ``program analysis'' when we see causal DAGs as our programs) to
read off these relationships directly from the graph itself through a structural
criterion that corresponds to conditional independence in the underlying
distribution. The notion of $d$-separation provides this graphical criterion.

\begin{definition}\label{def:d_conn} Let $\GG = (\VV, \EE)$ be a causal model.
    Then, nodes $a$ and $b$ are \textbf{\boldmath $d$-connected given $Z$},
    where $Z \subseteq \VV$, if and only if there exists an undirected path $P$
    from $a$ to $b$ in $\GG$ such that the following two conditions hold:
    \begin{enumerate}
        \item No mediator or confounder on $P$ is in $Z$.
        \item Every collider on $P$ has a descendant in $Z$ (we
        consider a node as its own descendant).
    \end{enumerate}
    We say $a$ and $b$ are \textbf{\boldmath $d$-separated given $Z$} if and
    only if they are not $d$-connected given $Z$.
\end{definition}
We can think of the variables in $Z$ as ``blocking'' any effect that $a$ may
have on $b$. For example, in \autoref{fig:caffeine}, $\verb|caffeine|$ and
$\verb|GPA|$ are $d$-separated given $Z = \{\verb|courseload|\}$ but not $Z =
\emptyset$.

In \textit{Causality}, \citet{pearl2009causality} connects the notions of $d$-separation and conditional
independence, stating that for a causal model $\GG$, if $a$
and $b$ are $d$-separated given $Z$, then conditional independence given $Z$
holds between $a$ and $b$ for every probability distribution $\mathbf{P}(\cdot)$
compatible with $\GG$. Conversely, if $a$ and $b$ are not $d$-separated given
$Z$, then there is at least one distribution $\mathbf{P}(\cdot)$ compatible with
$\GG$ in which $a$ and $b$ are not conditionally independent given $Z$.

While $d$-separation offers a graph-theoretic criterion for deciding conditional
independence, and the probabilistic definition formalizes it algebraically, both
leave a semantic gap: they do not explicitly connect the structure of a DAG to
the mechanisms that generate the observed distributions. 
On one hand, $d$-separation provides a syntactic rule for detecting conditional independence without grounding it
in the specifics of a data-generating process.
On the other hand, the probabilistic definition treats conditional independence as a property of
distributions without reference to underlying causal mechanisms, making it
agnostic to the data-generating process.
Neither approach fully captures the
intuition that conditional independence should arise from how variables are
causally related through structural assignments and shared randomness. We will
introduce a semantic definition to fill this gap by grounding conditional
independence in the functional dependencies encoded by the DAG.

More specifically, we have found the definition of $d$-separation unintuitive in
terms of the conditions it levies on paths. Why are these conditions \emph{exactly} the
right ones to rule out bad confounding? Pearl's proved equivalence~\cite{pearl2009causality} seems to
offer just that kind of answer, but much is hiding in the phrase ``compatible
with $\GG$.'' This statement restricts the probability distribution $\mathbf{P}$, such
that for every DAG node $v$ with parents $\pa(v)$ (those nodes with edges
heading into $v$), we have $\mathbf{P}(v) = \Pi_i \mathbf{P}(v \mid \pa(v)_i)$.
Setting aside
the difficulty of confirming such probability factorizations in practice, we note that conditional independence is a property of a particular distribution rather than of the graph structure itself. In particular, conditional independence can arise accidentally (e.g., through cancellation effects) even when a causal path exists. We are left wondering why the concept of probability is needed to explain
causality, which can also be formalized for purely deterministic systems.

To that end, we give a deterministic semantic characterization of what $d$-separation establishes, isolating its structural meaning before layering on probabilistic assumptions.
This characterization is not intended as a new algorithmic test; its purpose is to provide a faithful semantic foundation for $d$-separation, thereby justifying the use of $d$-separation as the computationally efficient static analysis of causal diagrams. 


\section{A New Semantic Characterization of Separation Under Conditions}
To capture the ideas of conditional independence in a more intuitive way, we introduce a new concept called \textit{semantic separation}: a characterization of $d$-separation that specifies how values propagate through a graph, rather than focusing on purely syntactic criteria. The semantic interpretation helps us understand not only the relationships among the nodes but also what these relationships imply about the behavior of variables under interventions or counterfactual scenarios.

\subsection{Function-Based Formal Semantics}\label{sec:function_based_formal_semantics}

\autoref{app:formalizing_causal_models} describes in detail how we represent causal models in Rocq, including the development of key graph-theoretic and causal-inference foundations.
In this section, we present the underlying semantic model in mathematical terms. The Rocq formalizations of the semantics are described in \autoref{app:semantics_impl}.

Causal models tell us the relationships that may exist between nodes. Assuming that a causal model accurately captures the causal relationships between nodes, each node's value is determined by a function of the values of its parents. In practice, however, the exact functional form of these relationships is unknown; to account for this uncertainty, we include an independent \textit{unobserved error term} for each node. These unobserved terms, common in causal inference, represent latent factors not explicitly included in the DAG. They are often visualized as greyed-out parent nodes. This abstraction reflects the reality that causal graphs rarely model all relevant background factors, so any omitted influences on a node are absorbed into its unobserved term. 

We choose to capture these relationships using node functions. A \textit{node function} assigns a value to a node based on its unobserved term and the values of its parents. Let $\mathcal{X}$ denote the domain of possible node values. Then each node
function takes an unobserved term in $\mathcal{X}$ and a list of parent values in $\mathcal{X}$ and
returns a value in $\mathcal{X}$. Note that the node function is of course free to ignore or allow little influence by the unobserved term. 

We represent a causal model with an overarching \textit{graph function}, which maps each node in the graph to its corresponding node function. Together with an assignment of unobserved terms, it determines the value of every node.

To illustrate how a graph function captures meaningful causal relationships, consider a simple example modeling how a student's sleep, study hours, and focus influence their test score. Suppose our causal graph is as shown in \autoref{fig:test_scores}.

\begin{figure}[b]
  \centering
  \newcommand{\figheight}{8cm}
  \newcommand{\figwidth}{2.0cm}
  \newcommand{\gap}{2em}

  \begin{tabular}{@{}c@{\hspace{\gap}}c@{\hspace{\gap}}c@{}}
    \raisebox{1em}{\begin{subfigure}[t]{0.25\textwidth}
      \centering
        \begin{tikzpicture}
    [
        node distance=1cm and 1cm,
        every node/.style={minimum size=1.2em, font=\small},
        circ/.style={draw, thick, circle, inner sep=2pt},
        dots/.style={font=\scriptsize}
        ]

    \node[] (C) {\textit{score}};
    \node[above left=.5cm of C] (B) {\textit{focus}};
    \node[above =.5cm of B] (A) {\textit{sleep}};
    \node[above right=.5cm of C] (L) {\textit{study}};
    
    \draw[->, thick] (A) -- (B);
    \draw[->, thick] (B) -- (C);
    \draw[->, thick] (L) -- (C);

    \end{tikzpicture}
    \end{subfigure}}
    &
    \raisebox{0em}{\begin{subfigure}[t]{0.15\textwidth}
      \centering
      \begin{tikzpicture}[
        node distance=1cm and 1cm,
        every node/.style={minimum size=1.2em, font=\normalsize},
        circ/.style={draw, thick, circle, inner sep=2pt}
    ]
    
    \node (x) at (2*.75,2*.75) {$x$};
    \node (r) at (0,2*.75) {$r$};
    \node[circ] (t) at (1*.75,1*.75) {$t$};
    \node (s) at (0,0) {$s$};
    \node (y) at (2*.75,0) {$y$};
    \node (u) at (1*.75,-1.2*.75) {$u$};
    
    \draw[->, thick] (x) -- (t);
    \draw[->, thick] (r) -- (t);
    \draw[->, thick] (t) -- (s);
    \draw[->, thick] (x) -- (y);
    \draw[->, thick] (s) -- (u);
    \draw[->, thick] (y) -- (u);
    
    \end{tikzpicture}
    \end{subfigure}}
    &
    \raisebox{9em}{\begin{minipage}[t]{.5\textwidth}
      \centering
      \caption{(Left) This causal model encodes the intuition that amount of sleep affects focus level, and test score is affected by both focus and study time.}
      \label{fig:test_scores}

      \caption{(Right) $Z = \{t\}$, and $\anc{u} = \{ u, s, y, x \}$. Note that although $x$ has a blocked path to $u$ through $t$, we only require the existence of any unblocked path.}
        \label{fig:unblocked_ancestors}
    \end{minipage}}
  \end{tabular}
\end{figure}

Then, we might have the following graph function $g$, where node values lie in $\mathcal{X} = \mathbb{R}$:

\[
g(v) =
\begin{cases}
(e) \mapsto e
& \textrm{if } v = \mathit{sleep} \\[4pt]

(e) \mapsto e
& \textrm{if } v = \mathit{study} \\[4pt]

(e, \mathit{sleep}) \mapsto g(\mathit{sleep}) + e
& \textrm{if } v = \mathit{focus} \\[4pt]

(e, \mathit{focus}, \mathit{study})
\mapsto 0.6\,g(\mathit{focus}) + 0.4\,g(\mathit{study}) + e
& \textrm{if } v = \mathit{score}
\end{cases}
\]

Here, $\mathit{sleep}$ and $\mathit{study}$ are exogenous variables whose values depend only on independent unobserved error terms $e$ representing factors such as stress, natural ability, distractions, or imprecision of measuring a variable. The variable $\mathit{focus}$ depends causally on $\mathit{sleep}$, and $\mathit{score}$ is a weighted combination of $\mathit{focus}$ and $\mathit{study}$, both also subject to their own unobserved terms.

Since each node's value depends on its parents' values, and the graph is assumed to be acyclic, the node functions can be evaluated in a topological order, yielding a value for every node and guaranteeing termination.

Different assignments of unobserved terms $U$ correspond to different
possible worlds consistent with the causal structure. Specifically, let $U: \VV \rightarrow \mathcal{X}$
be an assignment of unobserved terms to nodes (equivalently, we will sometimes represent $U$ or another assignment as a mapping, e.g., $\{ u: \alpha, v: \beta\}$).

Each node function can be computed with knowledge of $U$ and the structure of $\GG$, since it can then determine the unobserved term and parents of the desired node. Thus, we write $f_U : \mathcal{V} \rightarrow X$ as the graph function that encodes information about each node function and computes a node's value using $U$ as the unobserved-terms assignment for $\VV$.

For example, we would rewrite the above function $g$ as
\[
f_U(v) =
\begin{cases}
U(v)
& \textrm{if } v \in \{ \mathit{sleep}, \mathit{study} \} \\[4pt]
f_U(\mathit{sleep}) + U(v)
& \textrm{if } v = \mathit{focus} \\[4pt]
0.6\,f_U(\mathit{focus}) + 0.4\,f_U(\mathit{study}) + U(v)
& \textrm{if } v = \mathit{score}.
\end{cases}
\]

The function $f_U$ thus captures the entire semantics of the causal model under given assignments of unobserved terms.

\subsection{Defining Semantic Separation}
The semantic framework enables us to define an analogue of conditional independence in a way that aligns with intuition.

Informally, two nodes should be considered independent if changing the value of one has no effect on the value of the other. Under our model, the value of a node is determined by its function, so a natural semantic notion of independence would assert that altering the output of one node's function does not influence the output of the other's. We consider more generally \textit{conditional} independence, of which independence is the case where the conditioning set is empty. In this context, conditioning on a set of nodes $Z$ means we want to hold fixed the values of all nodes in $Z$ while assessing the influence of $u$ on $v$. To provide the fixed values, we introduce an additional set of assignments, $A_Z$, which maps each node in $Z$ to its desired conditioned value.
\begin{definition}\label{def:properly_conditions}
    We say that a world modeled by unobserved-terms assignments $U$ \textbf{properly conditions on \boldmath $Z$} if, for the given assignments $A_Z$, $f_U(z) = A_Z(z)$ for all $z\in Z$.
\end{definition}
In other words, $Z$ is the set of nodes being conditioned on, and $A_Z$ specifies the fixed values for those nodes. For example, if $Z = \{\text{gender}\}$ and $A_Z = \{\text{gender} : \text{F}\}$, then a world with function $f$ and unoberved-terms assignments $U$ properly conditions on $Z$ only if $f_U(\text{gender}) = \text{F}$.

Unlike probabilistic independence, at the level of functional semantics, each node's value depends deterministically on its parents and an unobserved error term. Our definition of semantic separation captures when such dependencies prevent causal influence from flowing between two nodes, given that certain variables must be preserved.

In particular, if $u$ and $v$ are semantically separated, then in a world where $f(u) = \alpha$ and all nodes in $Z$ are properly conditioned, if we were to modify the world \textit{minimally} so that $f(u) = \beta$ and all nodes in $Z$ remain properly conditioned, then $f(v)$ would be unaffected. The randomness or error leading to these different settings can be absorbed into the unobserved-terms assignments, $U$. 

We now dig into the meaning of ``minimally'': the goal is to ensure that the only possible influence on $f(v)$ comes from changes that causally descend from $u$. Since a node's value is affected only by its parents (and its own unobserved term), which are in turn only affected by their parents, we are led to restrict change to the \textit{ancestors} of $u$. However, since conditioned nodes have fixed values, changes to ancestors of $u$ can only affect $f(u)$ if they do not have to pass through conditioned nodes. We thus define the following notion.

\begin{definition}\label{def:unblockedancestor}
    Given a node $u$ and a conditioning set $Z$, a node $w$ is an \textbf{unblocked ancestor} of $u$ if either $w = u$, or there exists a directed path from $w$ to $u$ such that no internal nodes on the path (including $w$) are in $Z$. The set of unblocked ancestors of $u$ given $Z$ is denoted $\anc{u}$.
\end{definition}
An example is shown in \autoref{fig:unblocked_ancestors}.

When we intervene minimally on $u$, its change may propagate through the graph and
inadvertently alter variables in the conditioning set $Z$. Since variables in $Z$ must
remain fixed, we may need to adjust other variables so that the values of $Z$ are
restored. Intuitively, this means undoing, or \textit{repairing}, any changes to conditioned variables that arise from the intervention, in order to model the concept that $Z$ is being held constant structurally.

For example, consider determining how a drug influences health. Some influence may occur through blood-oxygen levels, while experimenters want to determine what other pathways exist. Therefore, they might compare subjects whose blood-oxygen levels are held fixed, counteracting the drug's effect on blood oxygen using some other well-established treatment.
This secondary intervention restores the conditioned variable. In general, however, such adjustments may themselves propagate through the system and disturb other factors that are intentionally being held constant, requiring further corrections until the conditioned variables are restored.

We formalize this intuition as a sequence of repair steps that restore the nodes in $Z$. With this setup, we now define semantic separation.
One way of seeing this definition is as defining a kind of \emph{operational semantics} of actions that experimenters might take, such that semantic separation holds when any such experimenter actions that reach certain termination conditions have preserved certain values.
\begin{definition}\label{def:ci}
    Let $\GG = (\VV, \EE)$ be a graph, and let $u,v \in \VV$. Let $f$ be a graph function, let $\alpha, \beta$ be values, and let $A_Z$ be conditioned assignments for $Z$. Let $U_0, U_1, ..., U_\ell$ be a \textit{sequence} of unobserved-terms assignments, which satisfy the conditions below:
    \begin{enumerate}
        \item \textit{Initialization.} $f_{U_0}(u) = \alpha$, and $f_{U_0}$ properly conditions on $Z$.
        \item \textit{Catalyst update.} $f_{U_1}(u) = \beta$, and $U_1$ differs from $U_0$ only for $a \in \anc{u}$.
        \item \textit{Reparative propagation.} For all $i = 2, ..., \ell$, $U_i$ differs from $U_{i-1}$ only for
        $$a' \in \bigcup_{\substack{z\in Z \\ \exists a \in \anc{z}, \\ U_{i-2}(a) \neq U_{i-1}(a)}} \anc{z}.$$
        \item \textit{Termination.} $1 \leq \ell \leq |\VV| + 1$.
        \item \textit{Restored conditioning.} $f_{U_\ell}(u) = \beta$, and $f_{U_\ell}$ properly conditions on $Z$.
    \end{enumerate}
    We say that $u$ and $v$ are \textbf{semantically separated given \boldmath $Z$} if, for all such $f, \alpha, \beta, A_Z, U_0, ..., U_\ell$, the value of $v$ remains unchanged:
    $f_{U_0}(v) = f_{U_\ell}(v).$
\end{definition}

\autoref{def:ci} is somewhat intricate, quantifying over many auxiliary assignments, but this complexity is intentional; the definition must faithfully capture the semantic idea that a change to $u$ must propagate through the model and, if it disturbs the conditioning set $Z$, can be repaired only by locally justified updates to unblocked ancestors of the disturbed nodes.

We now unpack the definition and build intuition for each requirement, showing that each condition is necessary to rule out spurious noncausal changes while allowing precisely those modifications that reflect legitimate causal propagation.

For instance, in the \textit{initialization}, we establish a reference world to anchor our semantics based on unobserved-terms assignments $U_0$, in which the value of $u$ is fixed at $\alpha$, and the system is well-behaved with respect to $Z$.

In the \textit{catalyst update}, we introduce the ``minimal modification'' to change the value of $u$ to $\beta$, yielding new unobserved-terms assignments $U_1$ that produce the new value for $u$. We need the additional constraint on the nodes for which $U_1$ differs from $U_0$ in order to isolate \textit{why} $f(v)$ changed in the case that the value of $v$ does not stay constant. Specifically, 
if we did not constrain $U_1$ at all, then a change in $f(v)$ might be attributable not to $f(u)$ but to other aspects of the changing unobserved-terms assignments. For example, in the simplest case, if $u$ and $v$ are disconnected nodes (and thus should be semantically separated for any $Z$), we could still have $U_0(v) \neq U_1(v)$, which could cause a difference in $f(v)$ that is entirely unrelated to $u$. Restricting changes to only unobserved terms of $\anc{u}$ ensures that when we compare the two different worlds represented by $U_0$ and $U_1$, those differences are localized to the nodes that could actually affect $u$. In doing so, we rule out irrelevant sources of variation.

At a high level, these first two stages simply describe how to introduce a perturbation to $u$ by changing only the unobserved terms of influence and observing whether the change affected $v$. To make the ideas more concrete, consider the following simple example.

\begin{example}
    Consider the case of the following graph: $x \rightarrow u \rightarrow y \rightarrow v$, where $Z = \emptyset$. Note that $u$ and $v$ are $d$-connected given $Z$, since the path $u\rightarrow y \rightarrow v$ is unblocked. Intuitively, we would expect that $u$ and $v$ are \textit{not} semantically separated, since $u$ has this direct causal path to $v$. 

    Suppose we begin with unobserved terms given by $U$. According to our definition, a change to the value of $u$ can be induced by modifying either $U(u)$ or $U(x)$ to get $U'$, since $f(u)$ could depend on both $U(u)$ and $f(x)$ (using the notation of \autoref{def:ci}, $U_0 = U$ and $U_1 = U'$). Such a change may propagate to $y$ (such that $f_U(y) \neq f_{U'}(y)$), whose value depends on $f(u)$, and then further to $v$, which depends on $f(y)$. Hence, a perturbation of $u$ can ultimately alter $f(v)$ (such that $f_U(v) \neq f_{U'}(v)$), showing that $u$ and $v$ are not semantically separated.
    
    Of course, there exist specific functions $f$ and unobserved-terms assignments $U$ for which perturbing $f(u)$ happens not to affect $f(v)$. However, \autoref{def:ci} quantifies over all such assignments and functions, requiring that no possible change to $u$ under the permitted conditions alter $v$. Therefore, the existence of even a single case where $f(v)$ changes is sufficient to conclude that $u$ and $v$ are not semantically separated.

    Now, consider the same graph but with $Z = \{y\}$. In this case, $u$ and $v$ are $d$-separated by $Z$. Returning to our definition, we see that for all $f$, $f(v)$ depends only on $U(v)$ and $f(y)$. Since $v$ is not an unblocked ancestor of $u$ or any conditioned node, $U(v)$ cannot change. Furthermore, since $y\in Z$, the \textit{restored conditioning} step requires that its value $f(y) = A_Z(y)$ stays fixed at the end of the process. Therefore, $f(v)$ remains invariant across all permissible changes, and $u$ and $v$ are semantically separated.
\end{example}






It may be tempting simply to define semantic separation by requiring that the value of $v$ stays constant through only the first two stages, since they capture the idea that the value of $u$ is perturbed, and $v$ must go unchanged. However, this approach fails to account for how values of $Z$ could affect $f(u)$; their effect will not be from the ancestors of $u$, since their passed-on values are fixed. This idea highlights a key insight: determining semantic separation requires not only a change to the value of $u$ but also the \textit{propagation} of that change throughout the rest of the graph. We are interested in whether the value of $v$ stays constant once the change has been fully propagated. In particular, after the initial catalyst change to $f(u)$, if any values of $Z$ are no longer properly conditioned, we allow changes to recondition them, leading to the stage of \textit{reparative propagation.}

The (possibly empty, since $\ell$ could equal $1$) sequence of unobserved-terms assignments $U_2, ..., U_\ell$ repairs any disrupted values in the conditioning set $Z$. In particular, for $i \geq 2$, each $U_i$ is responsible for restoring the conditioned nodes that were affected by the preceding assignment $U_{i-1}$. During the propagation process, previously repaired nodes may in turn disturb others, requiring further adjustments. Consequently, the repairs proceed gradually through the sequence, with each $U_i$ allowed to differ from $U_{i-1}$ only for ancestors of the nodes in $Z$ that were affected in the previous step. In this way, each set of assignments incrementally restores proper conditioning, using only changes justified by the preceding state.

\begin{example}\label{ex:def_collider}
    Consider the simple structure $u \rightarrow w \leftarrow v$ with $Z = \{ w \}$, as shown in \autoref{fig:def_collider}. We would expect $u$ and $v$ not to be semantically separated; for instance, consider $f(w) := f(u) \oplus f(v)$, which is permitted since $u$ and $v$ are both parents of $w$. Suppose $A_Z(w) = 1$, and we change the value of $u$ from $0$ to $1$. To preserve the conditioned value of $w$, the value of $v$ must change from $1$ to $0$. Thus, the values of $u$ and $v$ are dependent on each other. 

    We can observe this dependence via our definition. In this small graph, all catalysts must originate from changing $U(u)$ to obtain $U'$, since $u$ has no unblocked ancestors besides itself (using the concrete $f$ in the above paragraph, perhaps $f_U(u) := U(u) \text{ mod 2}$, $U(u) = 0$, $U'(u) = 1$). This change affects $f(w)$, so we must perform a reparative-propagation step to restore it. The repair could come from $U'(v)$ or $U'(w)$, since changing $U'(u)$ again would simply revert $u$ to its original value. (It is of course possible to have oscillating or cyclic sequences, repairing via $U'(u)$, which reperturbs $f(w)$, requiring yet another change, but here we consider the minimal sequence.) Suppose the change comes from $U'(v)$. As a result, $f(v)$ changes, showing that $u$ and $v$ are not semantically separated.
\end{example}
\begin{figure}
  \centering
  \newcommand{\gap}{1em}
  \begin{tabular}{@{}c@{\hspace{\gap}}c@{\hspace{\gap}}c@{}}
    \begin{subfigure}[t]{0.2\textwidth}
      \centering
        \begin{tikzpicture}[
    node distance=1.5cm and 1.5cm,
    every node/.style={minimum size=1.2em, font=\Large},
    circ/.style={draw, thick, circle, inner sep=2pt},
    unobs/.style={font=\small}
]

\node[red] (u) at (0, 1) {$u$};
\node[red, circ, below right=0.75cm of u] (w) {$w$};
\node[above right=0.75cm of w] (v) {$v$};

\node[unobs, gray, above=0.5cm of u] (Uu) {$U(u)$};
\node[unobs, gray, above=0.5cm of v] (Uv) {$U(v)$};
\node[unobs, gray] (Uw) at ($(Uu)!0.5!(Uv)$) {$U(w)$};
\fill[rounded corners=3pt, red, opacity=0.2] ($(Uu.south west)$) rectangle ($(Uu.north east)$);

\draw[->, thick, red] (u) -- (w);
\draw[->, thick] (v) -- (w);

\draw[->, red] (Uu) -- (u);
\draw[->, gray] (Uw) -- (w);
\draw[->, gray] (Uv) -- (v);

\end{tikzpicture}
    \end{subfigure}
    &
    {\begin{subfigure}[t]{0.2\textwidth}
      \centering
      \begin{tikzpicture}[
    node distance=1.5cm and 1.5cm,
    every node/.style={minimum size=1.2em, font=\Large},
    circ/.style={draw, thick, circle, inner sep=2pt},
    unobs/.style={font=\small}
]

\node[red] (u) at (0, 1) {$u$};
\node[blue, circ, below right=0.75cm of u] (w) {$w$};
\node[blue, above right=0.75cm of w] (v) {$v$};

\node[unobs, gray, above=0.5cm of u] (Uu) {$U(u)$};
\node[unobs, gray, above=0.5cm of v] (Uv) {$U(v)$};
\node[unobs, gray] (Uw) at ($(Uu)!0.5!(Uv)$) {$U(w)$};
\fill[rounded corners=3pt, blue, opacity=0.2] ($(Uw.south west)$) rectangle ($(Uv.north east)$);

\draw[->, thick, red] (u) -- (w);
\draw[->, thick, blue] (v) -- (w);

\draw[->, red] (Uu) -- (u);
\draw[->, gray] (Uw) -- (w);
\draw[->, blue] (Uv) -- (v);

\end{tikzpicture}
    \end{subfigure}}
    &
    \raisebox{8.5em}{\begin{minipage}[t]{.55\textwidth}
      \centering
      \caption{(Left) \textit{Catalyst update}: modifying $U(u)$ (highlighted) changes $u$, propagating to $w$ and violating its conditioned value. \\
      (Right): \textit{Reparative propagation}: restoring $w$ may use either $U(w)$ or $U(v)$ (highlighted). Here, repair is performed via $U(v)$, forcing a change in $v$.}
        \label{fig:def_collider}
    \end{minipage}}
  \end{tabular}
\end{figure}
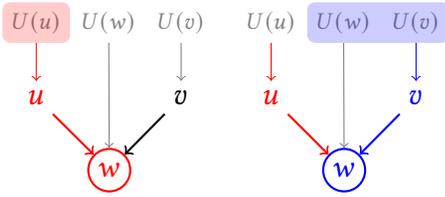
\autoref{ex:def_collider} aligns with the graphical intuition: $u$ and $v$ are $d$-connected in this graph since the collider $w$ is its own conditioned descendant. However, it is arguably more useful to reason through their lack of semantic separation, as above, than to observe the more-technical and less-intuitive collider-descendant rule given in $d$-separation.

\begin{example}
    Now consider a case with two consecutive conditioned nodes, such as the path $u \rightarrow z_1 \rightarrow z_2 \leftarrow v$, where $Z = \{ z_1, z_2 \}$, as shown in \autoref{fig:double_condition}. Note that $u$ and $v$ are $d$-separated because $z_1$ blocks the path between them as a conditioned mediator. 
    
    Suppose we change $U(u)$ to obtain $U'$, affecting $z_1$, which in turn may affect $z_2$. The reparative-propagation condition prevents an alternate repair of $z_2$ via $U'(v)$, since no \textit{unblocked} ancestors of $z_2$ changed from $U$ to $U'$. Instead, restoring $z_1$ (which \textit{does} have unblocked ancestors that changed) to its original value will automatically restore $z_2$.
    
    In this example, our formulation of the reparative-propagation step blocks a potential noncausal path of influence, ensuring that dependencies cannot ``skip through'' consecutive conditioned nodes. This restriction illustrates how the semantic definition enforces local, causally justified repairs.
\end{example}
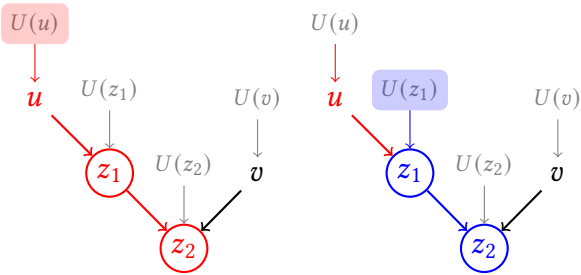
\begin{figure}[b]
  \centering
  \newcommand{\gap}{1.5em}
  \begin{tabular}{@{}c@{\hspace{\gap}}c@{\hspace{\gap}}c@{}}
    \begin{subfigure}[t]{0.25\textwidth}
      \centering
        \begin{tikzpicture}[
    node distance=1.5cm and 1.5cm,
    every node/.style={minimum size=1.2em, font=\Large},
    circ/.style={draw, thick, circle, inner sep=2pt},
    unobs/.style={font=\small}
]

\node[red] (u) at (0, 1) {$u$};
\node[red, circ, below right=0.75cm of u] (z1) {$z_1$};
\node[red, circ, below right=0.75cm of z1] (z2) {$z_2$};
\node[above right=0.75cm of z2] (v) {$v$};

\node[unobs, gray, above=0.5cm of u] (Uu) {$U(u)$};
\node[unobs, gray, above=0.5cm of v] (Uv) {$U(v)$};
\node[unobs, gray, above=0.5cm of z1] (Uz1) {$U(z_1)$};
\node[unobs, gray, above=0.5cm of z2] (Uz2) {$U(z_2)$};
\fill[rounded corners=3pt, red, opacity=0.2] ($(Uu.south west)$) rectangle ($(Uu.north east)$);

\draw[->, thick, red] (u) -- (z1);
\draw[->, thick, red] (z1) -- (z2);
\draw[->, thick] (v) -- (z2);

\draw[->, red] (Uu) -- (u);
\draw[->, gray] (Uz1) -- (z1);
\draw[->, gray] (Uz2) -- (z2);
\draw[->, gray] (Uv) -- (v);

\end{tikzpicture}
    \end{subfigure}
    &
    {\begin{subfigure}[t]{0.25\textwidth}
      \centering
      \begin{tikzpicture}[
    node distance=1.5cm and 1.5cm,
    every node/.style={minimum size=1.2em, font=\Large},
    circ/.style={draw, thick, circle, inner sep=2pt},
    unobs/.style={font=\small}
]

\node[red] (u) at (0, 1) {$u$};
\node[blue, circ, below right=0.75cm of u] (z1) {$z_1$};
\node[blue, circ, below right=0.75cm of z1] (z2) {$z_2$};
\node[above right=0.75cm of z2] (v) {$v$};

\node[unobs, gray, above=0.5cm of u] (Uu) {$U(u)$};
\node[unobs, gray, above=0.5cm of v] (Uv) {$U(v)$};
\node[unobs, gray, above=0.5cm of z1] (Uz1) {$U(z_1)$};
\node[unobs, gray, above=0.5cm of z2] (Uz2) {$U(z_2)$};
\fill[rounded corners=3pt, blue, opacity=0.2] ($(Uz1.south west)$) rectangle ($(Uz1.north east)$);

\draw[->, thick, red] (u) -- (z1);
\draw[->, thick, blue] (z1) -- (z2);
\draw[->, thick] (v) -- (z2);

\draw[->, red] (Uu) -- (u);
\draw[->, blue] (Uz1) -- (z1);
\draw[->, gray] (Uz2) -- (z2);
\draw[->, gray] (Uv) -- (v);

\end{tikzpicture}
    \end{subfigure}}
    &
    \raisebox{9.5em}{\begin{minipage}[t]{.4\textwidth}
      \centering
      \caption{(Left) \textit{Catalyst update}: modifying $U(u)$ (highlighted) changes $u$, propagating through $z_1$ to $z_2$ and violating both conditioned nodes. \\
      (Right): \textit{Reparative propagation}: restoring $z_1$ may only use $U(z_1)$ (highlighted), which also restores $z_2$.}
        \label{fig:double_condition}
    \end{minipage}}
  \end{tabular}
\end{figure}

The sequence continues until all conditioned nodes $z$ once again satisfy $f_{U_\ell}(z) = A_Z(z)$, after which we require that $v$ remain unchanged between the original and fully repaired worlds. This requirement is captured in \textit{restored conditioning}.

The \textit{termination} condition ensures that the definition is well-founded, and the sequence is finite. Since at least one conditioned node must be repaired for each $i$, we can bound the length of the repair sequence $U_2, ..., U_\ell$ by the number of conditioned nodes, which is of course upper-bounded by the total number of nodes $|\VV|$. As discussed in \autoref{ex:def_collider}, this sequence can be \textit{cyclic}, in that a repaired node becomes unconditioned again in a later set of assignments and must be repaired again. While cyclic sequences are not forbidden by \autoref{def:ci}, the definition must hold for all sequences satisfying the criteria, so redundant or oscillating repair paths can be ignored.

Semantic separation thus holds exactly when, across all such perturbation-repair scenarios, the value of $v$ remains invariant. Intuitively, changing $u$ cannot affect $v$ except through violations of $Z$, aligning directly with what $d$-separation expresses graphically. The next section will show that this semantic notion coincides precisely with $d$-separation, providing a formal bridge between graphical and semantic causal models.

\subsection{Semantic Separation \texorpdfstring{$\Longleftrightarrow$}{iff} \texorpdfstring{$d$}{d}-Separation} \label{sec:main_theorem}
Using the notion of semantic separation established in \autoref{def:ci}, we present our central result:
\begin{restatable}{theorem}{equiv}
\label{thm:equiv}
    The notions of semantic separation and $\mathrm{d}$-separation coincide exactly. In particular, for a causal model $\GG = (\VV, \EE)$, two distinct nodes $u,v \in \VV$, and a conditioning set $Z \subseteq \VV$ with $u, v \notin Z$, $u$ and $v$ are semantically separated given $Z$ (for all possible graph functions $f$ and conditioning assignments $A_Z$) if and only if they are $\mathrm{d}$-separated given $Z$ in $\GG$.
\end{restatable}

Our results are stated for an arbitrary domain $\mathcal{X}$ of node values. We assume:
\begin{itemize}
    \item $\mathcal{X}$ contains at least two distinct elements.
    \item Equality on $\mathcal{X}$ is decidable; that is, there exists a Boolean equality function satisfying reflexivity and symmetry.
\end{itemize}

We prove \autoref{thm:equiv} by splitting it into its two logical directions. We show the forward direction via the contrapositive: assuming that there exists a $d$-connecting path from $u$ to $v$, we define a specific graph function in which changing the value of $u$ alters the value of $v$, thereby violating semantic separation. We prove the backward direction by showing that if $u$ and $v$ are not semantically separated, i.e.,\ there exist unobserved-terms assignments satisfying the conditions of \autoref{def:ci} that propagate a change to the value of $v$, then we can show the existence of a $d$-connected path from $u$ to $v$. The forward and backward directions are proven in Sections \ref{sec:forward} and \ref{sec:backward}, respectively. 

Before turning to the practical advantages of this semantic view, it is helpful to clarify that our framework indeed assumes experimenters may carefully tweak only certain variables (e.g., the ancestors of $u$) within a world, while others remain fixed. This abstraction captures the idealized design of controlled experiments; while in practice such intervention may be approximate, many experimental setups, whether randomized trials in medicine or controlled simulations in physics, aim to perform precisely this kind of manipulation. In this sense, our semantics provides a natural target for how such experiments can be represented formally.

Assuming \autoref{thm:equiv} is true, we now examine situations in which the semantic definition offers benefits over the purely structural perspective offered by $d$-separation.

\section{Advantages of Semantic Separation}

\subsection{Special Cases of the Definition}

\autoref{def:ci} is not obviously more elementary than the standard definition of $d$-separation, which we called unintuitive.
However, specializing our new definition to different cardinalities of $Z$ yields simplified forms that one may find not much less intuitive than noninterference.

Consider first the case of $Z = \emptyset$, where it is easy to see that any repair step must be vacuous.
\begin{corollary}\label{cor:card0}
  Given a graph $\GG$ and two of its nodes $u$ and $v$, they are $\mathrm{d}$-separated without conditioning ($Z = \emptyset$) exactly when, under any compatible graph function (a graph function where each node's function depends only on its parents in the DAG), modifying any assignment $U$ only by changing \emph{the ancestors of $u$} can never modify the value of $v$.
\end{corollary}

This definition is extremely close to classic noninterference, with the ancestors of $u$ playing the role of the secret inputs and $v$ the role of the public output.
We have merely picked up a kind of thoroughness condition, forcing us to acknowledge every node upstream of $u$ as a potential influence.
(Classic noninterference is phrased in terms of black-box functions, with no concept of dataflow within, so it makes sense that the classic definition does not try to account for inputs causally influencing others.)

The true complexity of causal reasoning only becomes apparent when we need to account for confounding by choosing the right $Z$.
Let us consider the special case of $Z = \{z\}$.
\begin{corollary}\label{cor:card1}
  Given a graph $\GG$ and two of its nodes $u$ and $v$, they are $\mathrm{d}$-separated conditioning on $\{z\}$ exactly when, under any compatible graph function, modifying any assignment $U$ in either of the following ways can never modify the value of $v$.
  \begin{enumerate}
  \item Changing any ancestors of $u$ not blocked by $z$, such that the value of $z$ itself \emph{is not} affected.
  \item Changing any ancestors of $u$ not blocked by $z$, such that the value of $z$ itself \emph{is} affected; followed by a single step of repair, changing only ancestors of $z$ (since we assume a DAG, no relevant path to $z$ can go through $z$ internally, anyway) to reestablish the value of $z$ without changing the value of $u$.
  \end{enumerate}
\end{corollary}

\subsection{Test Oracles for Experiments}

The scientist's full repertoire depends on probability and statistics, which we have not yet folded into our formalization, but already we can justify some important ingredients in experiments.
Consider the general problem of determining whether a causal DAG is a correct formalization of a real phenomenon.
We may run experiments to try to \emph{disprove} the theory that we modeled the
phenomenon correctly with our DAG. Each run of an experiment measures some of
the nodes in G (partial $f_U$ for unknown $U$).

We want to know  \emph{when it is sound to declare the original DAG falsified,
considering particular measurements obtained in some experiment}.



\begin{corollary}
    Suppose we are given a graph $\GG$ and two of its nodes $u$ and $v$ that are
    $\mathrm{d}$-separated without conditioning (as can be confirmed via a simple graph
    algorithm). Consider an experimental run, implicitly running in some graph function
    $f$ and unobserved-terms assignments $U$, begun by measuring the values
    $f_U(u)$ and $f_U(v)$. Now intervene on u, producing $U'$, in some
    way guaranteed to change only the $U$ values for ancestors of $u$.  If
    measuring $v$ now confirms $f_U(v) \neq f_{U'}(v)$, then we have falsified the
    hypothesis that $\GG$ models the world accurately.
\end{corollary}
\begin{proof}
  By \autoref{cor:card0}.
\end{proof}


Now, the devil is in the details of how one intervenes during an experiment to ensure
that only the target variables are manipulated (formally captured by the do-operator, which replaces the intervened variable's function with a
constant so that it no longer depends on its usual parents~\cite{pearl2009causality}), while holding all other
structural relationships constant.
Ensuring that any real-world interventions do not change the unobserved factors of other nodes unintentionally is out of the scope of our framework.
Nevertheless, it is edifying to prove a principle for drawing conclusions from experiments
in a purely deterministic setting: $d$-separation can be useful even in the
absence of probability.

\begin{corollary}
  Suppose we are given a graph $\GG$ and two of its nodes $u$ and $v$ that are $\mathrm{d}$-separated
  conditioning on $\{z\}$.  Consider an experimental run, implicitly running in some graph
  function $f$ and unobserved-terms assignments $U$, begun by measuring the values
  $f_U(u)$, $f_U(v)$, and $f_U(z)$.  Now intervene on u, producing $U'$,
  in some way guaranteed to change only the $U$ values for ancestors of $u$ not
  blocked by $z$ in $\GG$, while also modifying the observed value of $z$.  Then
  make a further intervention into $U''$ guaranteed to change only the $U'$
  values for ancestors of $z$, restoring the prior value of $z$ without changing
  the value of $u$.  If measuring $v$ now confirms $f_U(v) \neq f_{U''}(v)$,
  then we have falsified the hypothesis that $\GG$ models the world accurately.
\end{corollary}
\begin{proof}
  By \autoref{cor:card1}.
\end{proof}

To illustrate these special cases, consider a simple example inspired by \autoref{fig:caffeine}, relating caffeine consumption to GPA. Suppose an experimenter hypothesizes a causal diagram in which extra caffeine has no direct effect on GPA once courseload is controlled. Our formalism for $|Z| = 1$ provides a way to test this hypothesis: the experimenter repeatedly recruits new subjects, measures their GPAs, and applies the intervention (e.g., increasing caffeine intake). Whenever a subject's courseload changes because of the intervention, we modify it directly, say through a requirement to drop excess classes. After such repairs, we confirm that GPA never changes. If it does, then the causal diagram -- and thus the scientific hypothesis -- is falsified.
In general, this procedure can be repeated to search systematically for counterexamples to the hypothesized causal model.

We have stated these theorems to cover what we might call \emph{soundness} of a test oracle for experiments: if the experimenter intervenes in a certain way and collects measurements satisfying a property, then we know that $\GG$ is inaccurate.

We might also aim for \emph{completeness}: intuitively, if $\GG$ is inaccurate, then some experiment of the given kind can be run to uncover that inaccuracy. 
The reality is not that simple, as \autoref{thm:equiv} describes semantic separation \emph{in all possible compatible graph functions $f$}, while a given sequence of experiments presumably only runs in the one $f$ that describes reality accurately.
That $f$ may have bad properties that prevent our chosen kind of experiment from noticing an inconsistency. 
We thus instead obtain completeness in terms of all compatible worlds, rather than a single realized world. 
Prior work effectively axiomatized the independence requirements of these worlds; our result shows that the semantic property can be proved without assuming such axioms.


Our hope is that our refined semantic understanding of $d$-separation can be a stepping stone to proving explicitly probabilistic theorems about soundness of experimental designs that appeal to checking relevant $d$-separations.
Having established the payoff of our central theorem, we return to its proof. We sketch the main ideas of each direction of implication, leaving most technical details to the appendices.

\section{Forward Direction: Semantic Separation Implies \texorpdfstring{$d$}{d}-Separation} \label{sec:forward}

We now prove the forward direction of \autoref{thm:equiv} by contrapositive. Suppose that $u$ and $v$ are not $d$-separated given $Z$, so there exists a $d$-connected path from $u$ to $v$ in $\GG$. We will show that $u$ and $v$ are not semantically separated given $Z$. Concretely, we will construct a graph function $f$ and unobserved-terms assignments $U_0, ..., U_\ell$ that satisfy the five conditions of \autoref{def:ci} but result in different values for $f_{U_0}(v)$ and $f_{U_\ell}(v)$, thus violating semantic separation.

At a high level, the strategy is as follows: we define a specific graph function $f$ that evaluates each node differently based on its structural role in the $d$-connected path, specifically whether it is a mediator, confounder, or collider. This function will be constructed such that all noncolliders on the path are assigned the same value. Since both $u$ and $v$ are endpoints and thus noncolliders, this property ensures that $f(u) = f(v)$. We will then construct a valid sequence of unobserved-terms assignments that modifies $u$'s value and causes the change to propagate through the path to $v$, ultimately changing $f(v)$ as well, which will demonstrate that $u$ and $v$ are not semantically separated and complete the contrapositive proof.

\subsection{Constructing a Function to Equate Node Values}
We examine simple cases of $d$-connected paths before generalizing to an arbitrary path.

\subsubsection{A Chain of Mediators}\label{sec:mediator_chain}
We start with the simplest case: the path from $u$ to $v$ consisting of a single edge $u \rightarrow v$. Since $u$ is a parent of $v$, we can define $f(v) := f(u)$ to propagate the value directly. Now suppose the path is a longer chain of mediators: $$u \rightarrow m_1 \rightarrow m_2 \rightarrow \cdots \rightarrow m_k \rightarrow v.$$
In this case, we define $f(v) := f(m_k)$, $f(m_i) := f(m_{i-1})$ for $i = 2, ..., k$, and $f(m_1) := f(u)$. Each node's value depends only on its parent in the path, and so this definition satisfies the structure of a causal model while ensuring all values along the path are equal.

We must also ensure that such a function respects conditioning on $Z$; otherwise it cannot be used in \autoref{def:ci}. In particular, if some descendant of a node on the path lies in $Z$, its value could be disrupted if the values along the path change. To prevent this improper conditioning, we simply define $f(z) := A_Z(z)$ for all $z\in Z$, ensuring that all conditioned nodes remain fixed and do not rely on any values of nodes in the path. Thus, we have shown that if $u$ and $v$ are connected by a path of all mediators (i.e.,\ a directed path), then there is a graph function that equates their values.

Note that the assumption that the path is $d$-connected is crucial. If any $m_i \in Z$, then changing $f(u)$ could violate the conditioning of $m_i$. Furthermore, the change would not propagate past $m_i$ to $v$, since $f(m_i) = A_Z(m_i)$ would be fixed.

The case of a directed path in the opposite direction (i.e.,\ towards $u$) is symmetric; we flip the function definitions accordingly to ensure again that all nodes evaluate to the same value.

\subsubsection{A Single Confounder}
We now consider a simple fork: $u \leftarrow c \rightarrow v$, where $c$ is a confounder. Since both $u$ and $v$ have $c$ as a parent, we can define $f(u) := f(c)$ and $f(v) := f(c)$. This graph-function definition will equate all values along the path to equal the value of the confounder.

More generally, suppose the path is two chains of mediators connected by a single confounder: $$u \leftarrow m_1 \leftarrow \cdots \leftarrow m_k \leftarrow c \rightarrow n_j \rightarrow \cdots \rightarrow n_1 \rightarrow v.$$
Then, we combine this idea with the ideas in \autoref{sec:mediator_chain}, defining the values of the left chain to be equal to $f(m_k)$ and the right chain to be equal to $f(n_j)$, and finally $f(m_k) := f(c)$ and $f(n_j) := f(c)$. Again, all values are equated along the path.

\subsubsection{A Single Collider}\label{sec:single_collider}
Suppose the path is $u \rightarrow c \leftarrow v$, where $c$ is a collider. We wish to equate all noncollider nodes on the path (just $u$ and $v$ in this case). This case is the least intuitive because $u$ and $v$ are both parents of $c$, but any effect of the value of $u$ on the value of $v$ must occur through $c$, their shared child node.

By the definition of $d$-connectedness, $c$ must have a descendant in $Z$. First, consider the case where $c \in Z$. Then, there is some assigned value $x := A_Z(c)$ which $c$ must evaluate to in every properly conditioned setting of unobserved-terms assignments. Choose some $y \neq x$, using the assumption that node values come from a type with at least two elements. Then, we can define
\[f(c) := \begin{cases}
    x & \text{if } f(u) = f(v) \\
    y & \text{otherwise},
\end{cases}\]
and any unobserved-terms assignments that properly condition on $Z$ necessarily force $f(u) = f(v)$.

Now, suppose $c\not \in Z$, so it has some descendant $d \in Z$. Let $d_1, ..., d_k$ be the intermediate nodes of the directed path from $c$ to $d$, as shown in \autoref{fig:single_descendant_path}.
Using the mediator-chain strategy of \autoref{sec:mediator_chain}, we can equate the values of nodes on $c$'s descendant path: $$f(d) := f(d_k), f(d_i) := f(d_{i-1}), f(d_1) := f(c),$$ so that $f(c) = f(d)$. We then define $f(c)$ in the same way as above, where $x := A_Z(d)$. Again, any unobserved-terms assignments that properly condition on $Z$ force $f(u) = f(v)$.

\begin{figure}[b]
  \centering
  \newcommand{\figheight}{8cm}
  \newcommand{\figwidth}{2.0cm}
  \newcommand{\gap}{2em}

  \begin{tabular}{@{}c@{\hspace{\gap}}c@{\hspace{\gap}}c@{}}
    \begin{subfigure}[t]{0.15\textwidth}
      \centering
        \begin{tikzpicture}[
    node distance=1cm and 1cm,
    every node/.style={minimum size=1.2em, font=\normalsize},
    circ/.style={draw, thick, circle, inner sep=2pt}
]

\node[blue] (u) at (0,1.5) {$u$};
\node[blue] (v) at (2, 1.5) {$v$};
\node[blue] (c) at (1,1) {$c$};
\node[below=.25cm of c] (d1) {$d_1$};
\node[below=.25cm of d1] (dots) {$\vdots$};
\node[below=.25cm of dots] (dk) {$d_k$};
\node[circ, below=.25cm of dk] (d) {$d$};

\draw[->, thick, blue] (u) -- (c);
\draw[->, thick, blue] (v) -- (c);
\draw[->, thick] (c) -- (d1);
\draw[->, thick] (d1) -- (dots);
\draw[->, thick] (dots) -- (dk);
\draw[->, thick] (dk) -- (d);

\end{tikzpicture}
    \end{subfigure}
    &
    \raisebox{2em}{\begin{subfigure}[t]{0.3\textwidth}
      \centering
      \begin{tikzpicture}[
    node distance=1cm and 1cm,
    every node/.style={minimum size=1.2em, font=\normalsize},
    circ/.style={draw, thick, circle, inner sep=2pt}
]

\node[blue] (u) {$u$};
\node[blue, right=.5cm of u] (t) {$t$};
\node[blue, right=.5cm of t] (q) {$q$};
\node[blue, right=.5cm of q] (r) {$r$};
\node[blue, right=.5cm of r] (v) {$v$};
\node[below=.75cm of t] (p) {$p$};
\node[circ, below=.75cm of p] (s) {$s$};
\node[below=.75cm of q] (y) {$y$};
\node[circ, right=1cm of s] (x) {$x$};

\draw[blue, ->, thick] (u) -- (t);
\draw[->, thick] (u) -- (p);
\draw[->, thick] (t) -- (p);
\draw[->, thick] (p) -- (s);
\draw[blue, ->, thick] (t) -- (q);
\draw[blue, ->, thick] (r) -- (q);
\draw[->, thick] (q) -- (y);
\draw[->, thick] (y) -- (x);
\draw[->, thick] (r) -- (x);
\draw[blue, ->, thick] (r) -- (v);

\end{tikzpicture}
    \end{subfigure}}
    &
    \raisebox{13.5em}{\begin{minipage}[t]{.4\textwidth}
      \centering
      \caption{(Left) Assuming the path highlighted in blue is $d$-connected, then $c$ must have a descendant in $Z$. Assuming $c\not\in Z$, it must have a descendant $d \in Z$ and a descendant path to $d$.}
      \label{fig:single_descendant_path}

      \caption{(Right) Consider the $d$-connected path from $u$ to $v$ highlighted in blue, where $Z = \{ s, x\}$. Then, the partition is as follows: $S_1 = \{u, r\}$, $S_2 = \{t, v\}$, $S_3 = \{q\}$, $S_4 = \{y, x\}$, $S_5 = \{s\}$, $S_6 = \{p\}$. }
        \label{fig:partition_example}
    \end{minipage}}
  \end{tabular}
\end{figure}

\subsubsection{General Construction of \texorpdfstring{$\fpath$}{fpath}} \label{sec:fpath_intro}
We now generalize the above ideas to an arbitrary $d$-connected path $P$ from $u$ to $v$. We construct a graph function $\fpath$ that forces all noncollider nodes on $P$ to take on the same value in a setting of unobserved terms that properly conditions on $Z$.

Specifically, we partition the nodes in $\VV$ into six sets and define a specific node function for each node based on the set into which it falls. Recall that each node function computes the node’s value from its unobserved term and the values of its parents.
\begin{itemize}
    \item $S_1$ (sources): nodes on $P$ with no neighboring parents on the path. The node function simply copies the value of the unobserved term.
    \item $S_2$ (transmitters): every node on $P$ with exactly one neighboring parent on the path. The node function simply copies the value of that parent along $P$.
    \item $S_3$ (colliders): every node on $P$ whose two neighboring nodes on the path are both parents. The node function compares the values of those two parents. If they are equal, then the node function outputs $x := A_Z(d)$,
    where $d$ is the conditioned descendant associated with the collider (possibly the collider itself). Otherwise, it outputs a distinct value $y \neq x$.
    \item $S_4$ (descendants): nodes along paths to conditioned descendants from colliders in $P$ (choosing one path per collider). The node function copies the value of the node's parent along the chosen descendant path.
    \item $S_5$ ($Z$-residual): conditioned nodes not included in $\bigcup_{i=1}^4 S_i$. The node function copies the node's assigned value in $A_Z$.
    \item $S_6$ (residual): remaining nodes. Assign any arbitrary node function.
\end{itemize}
Here, all confounders lie in $S_1$ and mediators in $S_2$. The endpoints $u$ and $v$ fall into either $S_1$ or $S_2$ depending on the directions of the edges adjacent to them. For example, using the graph in \autoref{fig:partition_example}, we would assign 
\begin{equation*}
    \begin{split}
        &\fpath(t) := \fpath(u), \hspace{2em}
        \fpath(v) := \fpath(r), \hspace{2em}
        \fpath(x) := \fpath(y) \\
        &\fpath(y) := \fpath(q), \hspace{2em}
        \fpath(q) := \begin{cases}
            A_Z(x) & \text{if } f(t) = f(r) \\
            \rho &\text{otherwise, for some $\rho \neq A_Z(x)$}.
        \end{cases}
    \end{split}
\end{equation*}




For each collider $c \not\in Z$, we select a descendant $d\in Z$ (which exists because $P$ is $d$-connected) and add all nodes on the directed path from $c$ to $d$ (excluding $c$, including $d$) to $S_4$.
However, an issue arises: we know that sets $S_1, S_2, S_3$ are disjoint since we assume the path from $u$ to $v$ is acyclic, and thus each node on the path is categorized into exactly one set. Then, we can assign each node in each set a different node function. However, nodes in $S_4$ may also be on $P$ or appear on multiple descendant paths, which would prevent the clean assignment of node functions.

Fortunately, whenever a $d$-connected path exists, there also exists one with all descendant paths disjoint from each other and from the path itself. The existence proof appears in \autoref{app:disjoint_desc_paths}. We can thus assume, without loss of generality, that the sets $S_1, ..., S_6$ form a disjoint partition of $\VV$, and we can thus fully construct $\fpath$. The Rocq definition is provided in \autoref{app:exset}.


We can see that under $\fpath$, the values of the nodes depend on the sources; each source will propagate its value to any neighboring transmitters. Thus, the values of all sources, determined by their unobserved terms, must agree.
\begin{definition}\label{def:source-fixed}
    For any $\alpha$, we say $U$ is \textbf{source-fixed to \boldmath $\alpha$} if $U(w) = \alpha$ for all sources $w \in S_1$.
\end{definition}
Given any source-fixed $U$, $\fpath_U$ indeed equates noncollider values. At a high level, a node in $S_2$ copies its parent's value, and chains of $S_2$ nodes are anchored by $S_1$ nodes, whose values are set directly by the unobserved terms, which are all equal since $U$ is source-fixed. Furthermore, since any two chains of transmitters that collide at an $S_3$ node take on the same value, the $S_3$ node will take on the correct $A_Z$ value to pass on to its conditioned descendant. Thus, for source-fixed $U$, $\fpath$ also properly conditions on $Z$. We formally show these two results in \autoref{app:equate_node_values}.

Thus, with source-fixed unobserved-terms assignments $U$, we have a function that properly conditions on $Z$ and forces the values of all noncollider nodes---most importantly, the values of $u$ and $v$---to be equal. We next construct a sequence of assignments that ensures the propagation of a value from $u$ to $v$, violating semantic separation.

\subsection{The Sequence of Unobserved-Terms Assignments}\label{sec:sequence_unobs}
We aim to use $\fpath$ to prove that $u$ and $v$ are \textit{not} semantically separated given $Z$, as described in \autoref{def:ci}. Choose some $\alpha \neq \beta$. If $U_0$ is source-fixed to $\alpha$, we will have $\fpath_{U_0}(u) = \fpath_{U_0}(v) = \alpha$. If we can create a sequence $U_1, ..., U_\ell$ satisfying the requirements of \autoref{def:ci} such that $U_\ell$ is source-fixed to $\beta$, then $\fpath_{U_\ell}(u) = \fpath_{U_\ell}(v) = \beta$, and notably $\fpath_{U_0}(v) \neq \fpath_{U_\ell}(v)$, proving that $u$ and $v$ are not semantically separated given $Z$.

We again examine simple cases of a $d$-connected path from $u$ to $v$ before providing the general construction of such a sequence.

\subsubsection{A Directed Edge}
In the case that $P$ is a single directed edge $u \rightarrow v$, then $S_1 = \{u\}$. We can simply define $U_0 = \{ u: \alpha \}$ and $U_1 = \{ u: \beta \}$, where the assignments for other nodes are arbitrary (but identical) for both $U_0$ and $U_1$. Note that $U_0$ and $U_1$ indeed differ only for members of $\anc{u}$ (only $u$). Here, $\ell = 1$, and $\fpath_{U_0}(v) \neq \fpath_{U_\ell}(v)$, so $u$ and $v$ are not semantically separated.

\subsubsection{A Single Confounder}
Now consider the case that $P$ is a simple fork $u \leftarrow w \rightarrow v$. Here, $S_1 = \{w\}$. Define $U_0 = \{ w : \alpha \}$, where other nodes are assigned arbitrarily. Note that changing the unobserved term of $u$ will not cause any change in $\fpath(u)$, since $\fpath(u)$ depends on the value of $w$. However, $w$ is an unblocked ancestor of $u$. Thus, we can define 
\[U_1(w') := \begin{cases}
    \beta & \text{if } w' = w \\
    U_0(w') & \text{otherwise}.
\end{cases}\]
Note that $U_0$ and $U_1$ are both source-fixed (to $\alpha$ and $\beta$, respectively), so if we let $\ell = 1$, we once again see that $U_0$ and $U_1 = U_\ell$ satisfy the requirements from \autoref{def:ci}, and thus $u$ and $v$ are not semantically separated.

\subsubsection{A Single Collider}
Now, consider the case that $P$ has a single collider: $u \rightarrow w \leftarrow v$, where $w\in Z$. Here, $S_1 = \{ u, v\}$. Define $U_0 = \{u: \alpha, v: \alpha\}$. Define
\begin{equation*}
    \begin{split}
        &U_1(w') := \begin{cases}
    \beta & \text{if } w' = u \\
    U_0(w') & \text{otherwise},
\end{cases} \hspace{2em} \text{and} \hspace{2em}
        U_2(w') := \begin{cases}
        \beta & \text{if } w' = v \\
        U_1(w') & \text{otherwise}.
    \end{cases}
    \end{split}
\end{equation*}
Note that under $U_1$, we are no longer properly conditioning on $Z$. In particular, $\fpath_{U_1}(u) = U_1(u) = \beta$, but $\fpath_{U_1}(v) = U_1(v) = \alpha$. Recall that $\fpath_{U_1}(w) = A_Z(w)$ only if $\fpath_{U_1}(u) = \fpath_{U_1}(v)$.

We thus perform a reparative-propagation step, defining $U_2$. Note that $U_2$ differs from $U_1$ for only $v \in \anc{w}$, which satisfies Condition 3 of \autoref{def:ci}. We let $\ell = 2$, and note that $U_\ell$ is source-fixed to $\beta$. Thus, $\fpath_{U_\ell} (v) = \beta \neq \fpath_{U_0}(v)$, so $u$ and $v$ are not semantically separated.

\subsubsection{General Construction of Sequence}
These simple cases show us that the propagation of a change in the value of $u$ can occur via the sources in the path (the nodes in $S_1$). In particular, if we change the unobserved terms of the sources, one-by-one, until the unobserved-terms assignments are source-fixed, and the value of $v$ is changed, then we can show that $u$ and $v$ are not semantically separated.

Concretely, suppose there are $\ell$ sources, and suppose they are organized in order of their appearance in $P$, such that $S_1 = [s_1, ..., s_\ell]$. Choose some $\alpha \neq \beta$. Define $U_0$ to be any unobserved-terms assignment source-fixed to $\alpha$. Then define the following sequence of unobserved-terms assignments, for $i = 1, ..., \ell$:
\begin{equation}\label{eq:sequence_U}
    \begin{split}
        U_i(w) = \begin{cases}
                \beta & \text{if } w = s_i \\
                U_{i-1}(w) & \text{otherwise}
            \end{cases}
    \end{split}
\end{equation}

Note that for all edge orientations of $P$, $S_1$ will have at least one node. Thus, the sequence will contain at least $U_1$, as needed for \autoref{def:ci}. 

The sequence also meets the remaining conditions of \autoref{def:ci}. In particular, the reparative-propagation step (Condition 3) holds because each triple of consecutive $U_i, U_{i+1}, U_{i+2}$ changes the unobserved terms for only two consecutive sources, which must be separated by a collider in the path, of which the two sources are both unblocked ancestors. The full proof is given in \autoref{app:forward}.

\section{Backward Direction: \texorpdfstring{$d$}{d}-Separation Implies Semantic Separation}\label{sec:backward}

We now prove the backward direction of \autoref{thm:equiv} by contrapositive:
assume that $u$ and $v$ are not semantically separated, so there exists some graph function $f$ and sequence of unobserved-terms assignments $U_0, ..., U_\ell$ satisfying the conditions of \autoref{def:ci}, such that $f_{U_0}(v) \neq f_{U_\ell}(v)$. We then use this sequence to show the existence of a $d$-connected path given $Z$ from $u$ to $v$, which will show that the two nodes are not $d$-separated.

\subsection{Change Originates From Unblocked Ancestors}
While in \autoref{sec:forward}, we leveraged a specific $d$-connected path to propagate a change to $v$, we now need to consider what a change in a function's value must imply about the structure of the graph.

In particular, recall that a node's value depends on its unobserved term and the values of its parents. Thus, if $f_{U}(v) \neq f_{U'}(v)$ for some $U, U'$, then it must be true that either $U(v) \neq U'(v)$, or $f_U(a) \neq f_{U'}(a)$ for some $a \in \pa(v)$. In the latter case, it again must be true that $U(a) \neq U'(a)$, or $f_U(a') \neq f_{U'}(a')$ for some $a' \in \pa(a)$. Since $\GG$ is acyclic, this chain of parents must eventually terminate at a node whose unobserved term in $U$ differs from its unobserved term in $U'$. Furthermore, note that each node in the chain is not in $Z$, since $f(z)$ is fixed for all $z\in Z$. 

We formally describe this conclusion in the following lemma:
\begin{restatable}{lemma}{valueaffectedbyunblocked}\label{lem:value_affected_by_unblocked}
    For any graph function $f$ and two unobserved-terms assignments $U, U'$, such that $f_U$ and $f_{U'}$ both properly condition on $Z$, if $f_U(w) \neq f_{U'}(w)$ for some $w\in \VV$, then there exists a node $a \in \anc{w}$ such that $U(a) \neq U'(a)$.
\end{restatable}
\begin{proof}
  See \autoref{app:backward}.
\end{proof}

\autoref{lem:value_affected_by_unblocked} already leads us towards a $d$-connected path to $v$ because it points us to a specific unblocked ancestor of $v$. In particular, if two nodes share an unblocked ancestor, then they are automatically $d$-connected, since we can enumerate the possible configurations that can occur: either the shared ancestor coincides with one of the two nodes, or a common unblocked ancestor acts as a confounder connecting the two nodes via two directed paths. In either case, all mediators are not conditioned on. In the second case, the confounder (the unblocked ancestor itself) is also not conditioned on. This result is shown formally in \autoref{app:backward}.


\subsection{\texorpdfstring{$d$}{d}-Connected Paths For Short Sequences}
We sketch the process of finding a $d$-connected path between $u$ and $v$ for short sequences of assignments that satisfy the conditions of \autoref{def:ci} but change the value of~$v$.



Suppose $\ell = 1$, so we change the value of $v$ between $U_0$ and $U_1$. Then, since $U_0$ and $U_1$ both properly condition on $Z$, \autoref{lem:value_affected_by_unblocked} tells us that there exists a node $a \in \anc{v}$ such that $U_0(a) \neq U_1(a)$. However, $U_0$ and $U_1$ are constrained such that they only differ for values in $\anc{u}$. Thus, $u$ and $v$ share an unblocked ancestor, so they must be $d$-connected.

Now, consider the case that $\ell = 2$, so $f_{U_0}(v) \neq f_{U_2}(v)$. Again, by \autoref{lem:value_affected_by_unblocked}, there exists a node $a \in \anc{v}$ such that $U_0(a) \neq U_2(a)$. If $U_1(a) = U_2(a)$, then by identical logic to the case $\ell = 1$, $u$ and $v$ are $d$-connected. Otherwise, if $U_1(a) \neq U_2(a)$, then there must exist $z \in Z$ such that $a \in \anc{z}$, and there is an $a'\in \anc{z}$ such that $U_0(a') \neq U_1(a')$. Then, $a' \in \anc{u}$. We can then construct the path shown in \autoref{fig:twoconfounderpatha}. We must consider the possibility that the directed paths making up the path overlap each other, which we handle formally in \autoref{app:backward_overlap}. One particularly interesting possible overlap is between the directed paths $a' \cdots z$ and $a \cdots z$. Note that if they do not overlap, then $z$ is a collider in the path, and $z \in Z$, so the path is indeed $d$-connected. It is however possible that the paths overlap at some node $c$ and take the same path to $z$. Luckily, this resulting path exactly mimics the setup of $d$-connectedness for a collider that has a descendant in $Z$, as shown in \autoref{fig:twoconfounderpathb}.
\begin{figure}
    \begin{subfigure}[b]{0.48\textwidth}
        \centering
        \begin{tikzpicture}[
        node distance=1cm and 1cm,
        every node/.style={minimum size=1.2em, font=\normalsize},
        circ/.style={draw, thick, circle, inner sep=2pt},
        dots/.style={font=\scriptsize}
        ]
        
        \node (a) at (1,1) {$a'$};
        \node[below left=.25cm of a, dots] (dots1) {\reflectbox{$\ddots$}};
        \node[below right=.25cm of a, dots] (dots2) {$\ddots$};
        \node[below left=.25cm of dots1] (w1) {$u$};
        \node[circ, below right=.25cm of dots2] (z) {$z$};
        \node[above right=.25cm of z, dots] (dots3) {\reflectbox{$\ddots$}};
        \node[above right=.25cm of dots3] (a1) {$a$};
        \node[below right=.25cm of a1, dots] (dots4) {$\ddots$};
        \node[below right=.25cm of dots4] (w2) {$v$};
            
        \draw[->, thick] (a) -- (dots1);
        \draw[->, thick] (dots1) -- (w1);
        \draw[->, thick] (a) -- (dots2);
        \draw[->, thick] (dots2) -- (z);
        \draw[->, thick] (a1) -- (dots3);
        \draw[->, thick] (dots3) -- (z);
        \draw[->, thick] (a1) -- (dots4);
        \draw[->, thick] (dots4) -- (w2);
        \end{tikzpicture}
        \caption{The two-confounder path constructed from unblocked ancestors $a'$ and $a$. Note that by construction, the path is $d$-connected.}
        \label{fig:twoconfounderpatha}
    \end{subfigure}
    \hspace{1em}
    \begin{subfigure}[b]{0.48\textwidth}
        \centering
        \begin{tikzpicture}[
        node distance=1cm and 1cm,
        every node/.style={minimum size=1.2em, font=\normalsize},
        circ/.style={draw, thick, circle, inner sep=2pt},
        dots/.style={font=\scriptsize}
        ]
        
        \node (a) at (1,1) {$a'$};
        \node[below left=.25cm of a, dots] (dots1) {\reflectbox{$\ddots$}};
        \node[below right=.01cm and .25cm of a, dots] (dots2) {$\ddots$};
        \node[below left=.25cm of dots1] (w1) {$u$};
        \node[below right=.01cm and .25cm of dots2] (c) {$c$};
        \node[below=.2cm of c, dots] (vdots) {\raisebox{.1cm}{$\vdots$}};
        \node[circ, below=.2cm of vdots] (z) {$z$};
        \node[above right=.01cm and .25cm of c, dots] (dots3) {\reflectbox{$\ddots$}};
        \node[above right=.01cm and .25cm of dots3] (a1) {$a$};
        \node[below right=.25cm of a1, dots] (dots4) {$\ddots$};
        \node[below right=.25cm of dots4] (w2) {$v$};
            
        \draw[->, thick] (a) -- (dots1);
        \draw[->, thick] (dots1) -- (w1);
        \draw[->, thick] (a) -- (dots2);
        \draw[->, thick] (dots2) -- (c);
        \draw[->, thick] (a1) -- (dots3);
        \draw[->, thick] (dots3) -- (c);
        \draw[->, thick] (a1) -- (dots4);
        \draw[->, thick] (dots4) -- (w2);
        \draw[->, thick] (c) -- (vdots);
        \draw[->, thick] (vdots) -- (z);
        \end{tikzpicture}
        \caption{If the paths overlap at collider $c$, then the path is still $d$-connected since $z$ is a conditioned descendant.}
        \label{fig:twoconfounderpathb}
    \end{subfigure}
    
    \caption{For any $u$, $v$ such that $v$'s value is affected by a change in $u$ with a sequence of unobserved-terms assignments $U_0, U_1, U_2$, we can construct a $d$-connected path from $u$ to $v$ using $z\in Z$ and shared ancestors $a', a$.}
\end{figure}
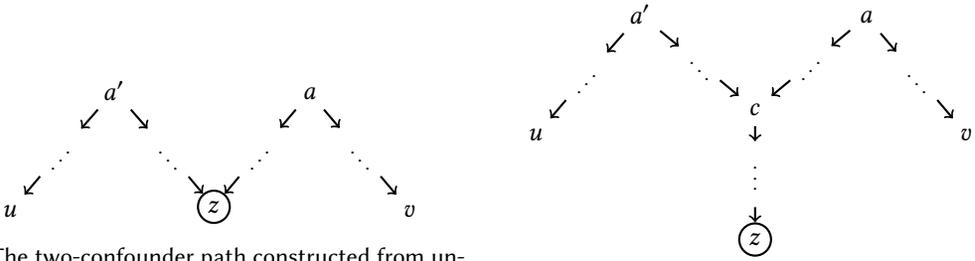

Thus far we have only constructed examples involving paths with at most one collider and two confounders. To establish full equivalence to $d$-separation, we expect to rely on $d$-connected paths of arbitrary structure with any number of mediators, confounders, and colliders.
This observation provides intuition behind why our definition of semantic separation requires a \textit{sequence} of assignments rather than a single step: the change must propagate through potentially many intermediate nodes, depending on the structure of $\GG$. In the following section, we generalize this reasoning to arbitrary-length sequences that satisfy the constraints in \autoref{def:ci}, thereby leveraging the full expressive power of $d$-connectedness.

\subsection{Generalizing to Arbitrary-Length Sequences}\label{sec:arb_length}
Intuitively, for $i \geq 1$, each transition from $U_i$ to $U_{i+1}$ changes an unblocked ancestor whose influence passes through a shared descendant $z \in Z$. Each transition introduces a collider, and concatenating the directed paths from these ancestors to their descendants yields a $d$-connected path, showing that $u$ and $v$ must be $d$-connected.

A natural strategy for generalizing this logic to arbitrary-length sequences $U_0, ..., U_\ell$ is to attempt induction on $\ell$. However, this approach does not directly succeed.
In particular, let the induction hypothesis be that any sequence with $\ell = n$ showing semantic dependence between two nodes implies a $d$-connected path between them.
For a sequence with $\ell = n+1$, one might try applying the hypothesis to the subsequence $(U_1, U_2, \dots, U_{n+1})$ to get a path from some $z \in Z$ to $v$, then prepend a path from $u$ to $z$.
However, this approach fails because $U_1$ may differ from $U_0$ only for ancestors of $u$, while later steps may involve changes across ancestors of multiple $z_i \in Z$. Since there is no guarantee of a $d$-connected path from each of these $z_i$ to $v$, we cannot inductively extend the argument.
Applying the hypothesis to the prefix $(U_0, U_1, \dots, U_n)$ encounters a different problem: $U_{n+1}$ must fully recondition on $Z$, whereas intermediate steps need not, which once again prevents a clean inductive step. 

Importantly, the sequence need not be minimal: $v$ may change earlier than the final step. If $v$ changes after $U_1$, some $z \in Z$ must have changed first, enabling $v$’s downstream change.
This observation motivates the correct inductive strategy. In the base case, $v$ changes directly due to $U_1$, implying a shared unblocked ancestor with $u$. Otherwise, $v$ changes through a node $z \in Z$, and by locating the step where $z$ was first affected, we can apply induction: either $z$ itself shares an unblocked ancestor with $u$, or its change was triggered by the repair of another node $z'$, recursively extending the path toward $u$.
Full details are provided in \autoref{app:cons_path}.

By concatenating this inductively obtained path from $u$ to $z$ with the path from $z$ to $v$, we construct a $d$-connected path from $u$ to $v$. Full details of this construction and its role in completing the backward direction appear in \autoref{app:backward_overlap}, thus establishing the central equivalence.

\section{Related Work} \label{sec:relatedWork}

We build on the formalism of causal diagrams that \citet{Pearl1995CausalDiagrams,pearl1988probabilisticReasoning} introduced.
Geiger, Verma, and Pearl~\cite{geiger1990dSeparation} first established the equivalence between $d$-separation and a semantic notion based on three types of independence: probabilistic, relational, and correlational. 
We offer a deterministic perspective. 
A book~\cite{pearl2009causality} and a survey article~\cite{Pearl2010Introduction} are helpful to overview more-recent developments.
The semantics style that we use in this paper, where graph nodes are given meanings via deterministic functions, comes from viewing causal diagrams as encoding information about conditional independences~\cite{pearl1988probabilisticReasoning}.

\subsection{Program Analysis for Scientific Validity}
PlanOut~\cite{bakshy2014planout} is a DSL for specifying and executing large-scale online experiments.
Planalyzer~\cite{tosch2019planalyzer} checks internal-validity violations in PlanOut scripts and generates statistical-analysis plans for estimand calculations.
These projects are not grounded in world models, while we aim to develop foundations for such grounding through semantics of causal networks.
Similarly, Dagitty~\cite{textor2011dagitty} enables graphical causal analysis but requires manual alignment by users of experiments and statistical analyses with graphs, whereas we aim to build towards formally proving such correspondences.

Research on synthesizing statistical analyses spans HCI and statistics
but is often domain-specific (e.g., Experiscope~\cite{guimbretiere2007experiscope} for sensor logs, Statsplorer~\cite{wacharamanotham2015statsplorer} for teaching statistics) or 
supports limited data-collection strategies (e.g., Touchstone~\cite{mackay2007touchstone}, Touchstone2~\cite{eiselmayer2019touchstone2}). 
The Automatic Statistician~\cite{lloyd2014automatic} is a general-purpose approach but relies solely on data, compromising external and internal validity.

The Tea~\cite{jun2019tea}, Tisane~\cite{jun2022tisane}, and rTisane~\cite{jun2024rtisane} projects are most closely related to our work, as they
(i) assess statistical-analysis synthesis through the lens of scientific validity and
(ii) introduce a formal-methods perspective to reasoning about validity. However, their focus is on establishing that statistical models can be synthesized from input causal graphs in a way that is usable by end-users who are not causality experts. As such, these systems feature higher-level DSLs for eliciting domain-specific causal assumptions and apply recommendations from the statistics literature to construct regression models. They neither reason about the semantic properties of causal diagrams nor reconcile functional and graph-theoretic views of causality. Instead, their implementations operate on graph patterns (e.g., identifying mediators, confounders, and colliders) without specifying what these structures mean semantically about the assumed world.

Our work focuses on establishing a clear semantic foundation of causality. In our envisioned scientific reasoning stack, our work provides lower-level foundations on which systems such as Tea, Tisane, and rTisane could build. In the future, such systems could leverage our framework to obtain provably valid guarantees rather than relying solely on informal statistical ``best practices.''

\subsection{Probability in Semantics and Verification}

One major use of probability in programming languages is in probabilistic programming, introduced with the Church programming language~\cite{Goodman2008Church} and remaining an active area of study in the programming-languages research community (e.g. with Anglican~\cite{Tolpin2016Anglican}) and statistical practice (e.g. with Stan~\cite{Carpenter2017Stan}).
\citet{Kozen1981Semantics} gave an early formal semantics of probabilistic programs, \citet{McIverMorgan2005pGCL} developed program-verification methods in this domain, and a variety of verification work has followed, for instance in relational reasoning for privacy properties~\cite{Barthe2016Relational}.
These two bodies of work use probability explicitly to state their results, while our work shows how a celebrated probabilistic theorem can also be given a deterministic characterization.

\subsection{Program Analysis and Information-Flow Security}

Abstract interpretation~\cite{AbsInt} is an influential early formalism, assigning meanings to summaries of program state computed by dataflow analysis~\cite{Kildall1973Dataflow} with uniform proof that those summaries are accurate.
A starting point for our work was trying to attach similar \emph{semantic} correctness conditions to forms of analysis from causal inference, starting with $d$-separation.

The closest such semantic condition from the literature is noninteference as proposed by \citet{noninterference}, creating the basis for almost all later work on information-flow security.
Dual versions of noninterference capture both confidentiality (an adversary does not learn secret values of variables) and integrity (an adversary cannot influence the values of important variables).
An especially well-developed proof approach is \emph{security types}~\cite{JFlow}, which e.g., associate program variables with levels of secrecy and use static analysis to check that the types are respected.
Similar guarantees can also be provided through run-time analysis, as developed in both operating systems~\cite{Zeldovich2006HiStar} and high-level programming languages~\cite{Stefan2011LIO}.

Our work enables an analogy between, on the one hand, security types and noninterference; and, on the other hand, $d$-separation and our new semantic characterization.
Quantitative information-flow security~\cite{QuantitativeIflow} has also been studied for some time in the context of probabilistic state machines and offers inspiration for our planned extensions toward probability and statistics.

\section{Limitations and Future Work}

Our goal is to develop formal foundations for reasoning about causal diagrams
as part of a broader theory of experiment verification. In particular,
semantic separation captures a basic task in scientific reasoning:
determining when one variable cannot causally influence another under
assumed relationships. In this sense, our work isolates a
\emph{structural} layer of reasoning about causality.

Future work should build up to the explicitly probabilistic and statistical
setting needed to characterize real experiments.
The ``full stack'' of scientific reasoning incorporates not only probability and
statistics (e.g., assuming a normal distribution for each unobserved quantity)
but also further assumptions about possible $f$ values (e.g., perhaps every node
function must be linear). We see several key ingredients of scientific analysis
yet to be justified through future work.
\begin{itemize}
\item Falsifying graphs as world models is a good start, but in general we want
to run experiments to narrow the sets of possible graph functions.  For
instance, if we assume graph functions are linear, we can learn tighter
intervals containing their true coefficients.
\item Real measurements often have continuous domains, making any principle
questionable when it appeals to precise equality or inequality, motivating use
of assumptions about probability distributions behind common statistical
methods.
\item As mentioned above, it may be difficult to convince ourselves that
experimental procedures only modify unobserved terms relevant to particular
graph nodes, or we may need to run more-observational studies that depend on
sampling subjects with particular inherent properties.  Randomized selection of
subjects and assignment to conditions (e.g., treatment vs. control) then become
important. 
\end{itemize}

We are optimistic that the foundations we have begun to develop 
will help characterize the building blocks of future probabilistic frameworks in that vein.
Today, PLDI papers often include mechanized proofs of their key claims, and perhaps some day papers throughout the natural and social sciences can have associated mechanized proofs of soundness for their experiment designs.

\section*{Artifact-Availability Statement}
The Rocq development \cite{zhang_2026_19075958} was evaluated and is publicly available.

\bibliographystyle{ACM-Reference-Format}
\bibliography{paper,emjun}

\newpage

\appendix

\section{Formalizing Causal Models in Rocq}\label{app:formalizing_causal_models}
This appendix presents the foundational work of implementing causal diagrams within the Rocq Prover. We begin by implementing DAGs and build upon this framework to formalize key causal concepts and theorems, combining classic graph theory with domain-specific ideas from causal inference.

\subsection{Causal Diagrams as Directed Acyclic Graphs}
Representing DAGs in Rocq involves designing data structures to represent directed graphs and implementing associated functions, including pathfinding, cycle detection, and topological sort. Importantly, each of these functions should have its correctness formally proven using Rocq, ensuring that it performs as expected under all circumstances.

We choose to represent causal models in terms of \verb|nodes| and \verb|edges|:
\begin{minted}{coq}
Definition node: Type := nat.
Definition nodes := list node.

Definition edge: Type := node * node.
Definition edges := list edge.

Definition graph: Type := nodes * edges.
\end{minted}
Nodes are simply natural numbers. These choices allow us to reuse several \verb|nat|-specific functions and theorems that are available in the Rocq standard library, while still representing nodes in an intuitive way. 

Paths are also a crucial part of causal models, since $d$-connectedness requires the existence of a $d$-connecting path. We represent paths as follows:
\begin{minted}{coq}
Definition path: Type := node * node * nodes.
Definition paths := list path.
\end{minted}
Here, a path is represented as a tuple of its start node, end node, and intermediate nodes. For example, the path $(1, 4, [2;3])$ would represent the path $1 \rightarrow 2 \rightarrow 3 \rightarrow 4$. This choice of representation allows us to access the endpoints of a path easily as well as require that a path has a start and an end node. Note that by this definition, a path has at least one edge; the path $(1, 1, [])$ would denote the self-loop at node 1, rather than the 1-path consisting only of node 1.
We often use more standard notation and write a path $(u, v, l)$ as a list $[u] {++} l {++} [v]$ instead.

Formalizing basic graph operations in Rocq is far from straightforward; the purely functional setting and the need for total, terminating programs mean that even relatively simple procedures require careful construction. 
In our development, standard graph-theoretic operations, including pathfinding, cycle detection, and topological sorting, have been mechanized and verified. This includes formal proofs of classic results that are often assumed in pen-and-paper treatments, such as the existence of an enumeration of all acyclic paths between two nodes and the property that parents precede children in a valid topological sort. The correctness of the system thus follows directly from the verified Rocq development.\footnote{At the time of artifact submission, the development is fully self-contained, with functional extensionality being the only remaining axiom.}


We highlight two major categories of graph functionality: pathfinding and topological sorting. Their implementation demonstrates how complex graph-theoretic properties can be handled within Rocq’s purely functional and constructive framework.

\subsubsection{Pathfinding} \label{sec:pathfinding}
Our pathfinding implementation relies on recursive exploration. Given a graph $\GG$ and two nodes $u,v$, we aim for the function $\verb|find_all_paths_from_start_to_end|(u, v, \GG)$ to return a list of paths that contains all undirected paths between $u$ and $v$ in $\GG$. The function relies on numerous intermediate results. 
\begin{itemize}
    \item \begin{minted}{coq}
Fixpoint edges_as_paths_from_start (u: node) (E: edges): paths
    \end{minted}
    Returns list of paths of length 2 that start from $u$ by transforming edges that are incident to $u$ into paths, i.e., $(u, v) \mapsto (u, v, [])$. 
    \item \begin{minted}{coq}
Fixpoint extend_paths_from_start_by_edge (e: edge) (l: paths): paths
    \end{minted}
    Returns list of paths that contains all paths in \verb|l|, as well as all paths that can be constructed from extending a path in \verb|l| by \verb|e|. An extension is only performed if \verb|e| is incident to the path's endpoint and the resulting path does not contain a duplicate node, thereby preventing the introduction of cycles into the path set.
    \item \begin{minted}{coq}
Fixpoint extend_paths_from_start_by_edges (E: edges) (l: paths): paths
    \end{minted}
    Returns result of repeatedly calling \verb|extend_paths_from_start_by_edge| for all edges in \verb|E| on the result of the previous call, beginning with \verb|l|.
    \item \begin{minted}{coq}
Fixpoint extend_paths_from_start_iter
                (E: edges) (l: paths) (k: nat): paths
    \end{minted}
    Returns result of repeatedly, \verb|k| times, calling \verb|extend_paths_from_start_by_edges| on the result of the previous call, beginning with \verb|l|.
\end{itemize}
Then, given a graph $\GG = (\VV, \EE)$, a start node $u$, and an end node $v$, we can find all undirected and acyclic paths from $u$ to $v$ in $\GG$ by iterating \verb|extend_paths_from_start_iter| $|\VV|$ times:
\begin{minted}{coq}
Definition find_all_paths_from_start_to_end
                            (u v: node) (G: graph): paths :=
  match G with
  | (V, E) => filter (fun p => v =? path_end p)
                (extend_paths_from_start_iter E
                    (edges_as_paths_from_start u E) (length V))
  end.
\end{minted}

While the resulting function is inefficient compared to typical imperative algorithms, these definitions are intended primarily to be used as part of a formal semantics rather than for execution, and they are furthermore sufficient for the relatively small causal models used in practice.

The correctness of this implementation is established by a mechanized theorem, which asserts that a path $p$ is in the result of \verb|find_all_paths_from_start_to_end| if and only if it is a valid, acyclic, undirected path from $u$ to $v$. The proof is structured into two main parts:

\begin{itemize}
    \item \textit{Soundness:} We prove by induction on the number of iterations $k$ that every path generated by \verb|extend_paths_from_start_iter| is a valid path in the graph and contains no repeating nodes. This relies on the check in \verb|extend_paths_from_start_by_edge|, which explicitly rejects extensions that would create a cycle. Since the final step filters these results for paths ending in $v$, any returned path is guaranteed to be a valid acyclic path from $u$ to $v$.
    
    \item \textit{Completeness:} We show that any acyclic path of length $n$ in the graph will be captured by the $n$-th iteration of the extension process. Since an acyclic path in a graph $(\VV, \EE)$ can have at most $|\VV|$ nodes, iterating $|\VV|$ times is sufficient to cover the entire space of possible acyclic paths. The proof uses induction on the structure of a given arbitrary path $p$ to show that each prefix of $p$ is present in the corresponding step of the iterative extension.
\end{itemize}

Building on the pathfinding routine, we can define a function to detect directed paths, which are essential for identifying descendants and ancestors of nodes, both of which are concepts used extensively in our semantic treatment of causal models. We also define a function to detect cycles, which is a necessary safeguard in many downstream functions such as topological sort and tests for $d$-separation.

\subsubsection{Topological Sort}\label{sec:topo_sort}
Topological sorting plays a central role in causal inference, providing an ordering of variables consistent with the direction of causality. This ordering is also crucial for defining the semantics of causal models, which rely on evaluating each node's value based on the values of its parents.

In Rocq, our function for topological sort must account for the possibility that the inputted graph is \textit{not} acyclic. To this end, the return type is wrapped in an \verb|option|; the function returns \verb|Some ts| if a valid topological sort \verb|ts| exists and \verb|None| otherwise:
\begin{minted}{coq}
Definition topological_sort (G: graph): option nodes
\end{minted}

Our implementation is based on indegree counting. We iterate $|\VV|$ times, where $\GG = (\VV, \EE)$ is the inputted graph. In each step, we find a node with indegree zero (a node with no parents), append it to the resulting list, and remove it and all its incident edges from the graph. The algorithm recurses on the remaining subgraph until either all nodes are processed or no indegree-zero node is found (in which case the graph contains a cycle and the function returns \verb|None|).

With this function defined, we prove several important properties, such as the sort containing exactly the graph's nodes with no duplicates. Most notably, we mechanize the existence guarantee: that the function always returns some valid topological order for an acyclic input. While this is a classic result, the mechanized proof requires induction on the number of nodes and several auxiliary lemmas, including that acyclicity is preserved at each step and an application of the pigeonhole principle to derive a contradiction.

\subsection{Formalizing Key Causal Concepts}
With the basic DAG framework established, we shift focus to more specific elements of causal models, such as classifying intermediate nodes in a path and determining $d$-separation and $d$-connectedness.

\subsubsection{Identifying Mediators, Confounders, and Colliders}\label{sec:med_con_col_functions}
We begin by defining functions to extract mediators, confounders, and colliders from a given path. For instance, to identify colliders, we define a function that, given a path $(u, v, l)$ and a graph $\GG$, iterates through the node list $[u] {++} l {++} [v]$ and inspects each triplet of consecutive nodes. If the triplet satisfies the criteria for a collider, namely, the two edges are both incoming into the center node, we append the center node to the result. The output is a list of all colliders in the path.

However, the fact that a node is a collider in a path is much less useful than knowing which neighboring nodes make it a collider. We thus prove the following theorem:
\begin{minted}{coq}
Theorem colliders_vs_edges_in_path: forall (L: nodes) (G: graph) (x: node),
    In x (find_colliders_in_nodes L G)
    <-> exists y z: node, sublist [y; x; z] L = true
            /\ is_edge (y, x) G = true /\ is_edge (z, x) G = true.
\end{minted}
Here, the input \verb|L| would be the full list of nodes in a path, i.e., $[u] {++} l {++} [v]$. This result allows us to extract the specific triplet in the path of which the node is a collider, which is crucial in more involved reasoning about these structures.

We also prove useful structural lemmas, including that the set of mediators, confounders, and colliders is preserved under path reversal and that every internal node of a directed path is a mediator. These are proven by induction on the length of the path.

Importantly, we show that in acyclic graphs and acyclic paths, no node can simultaneously play more than one of these structural roles. The specific case for mediators is shown below:
\begin{minted}{coq}
Theorem if_mediator_then_not_confounder_path:
    forall (G: graph) (u: node) (p: path),
        contains_cycle G = false /\ acyclic_path p
            /\ In u (find_mediators_in_path p G)
        -> ~In u (find_confounders_in_path p G)
           /\ ~In u (find_colliders_in_path p G).
\end{minted}
This theorem is proven by case analysis, showing that the structural criteria for the three categorizations are mutually exclusive in an acyclic setting.

These foundational results are critical for building the semantics of causal inference, which rely on mediators, confounders, and colliders to characterize paths.

\subsubsection{Determining \texorpdfstring{$d$}{d}-Separation and \texorpdfstring{$d$}{d}-Connectedness}
With the structural elements defined, we implement a function to determine whether a path is blocked by a conditioning set $Z$. We define three separate functions that test whether a given path is blocked by a mediator, confounder, or collider, and we combine them into a single predicate:
\begin{minted}{coq}
Definition path_is_blocked_bool (G: graph) (Z: nodes) (p: path): bool :=
    is_blocked_by_mediator p G Z || 
    is_blocked_by_confounder p G Z || 
    is_blocked_by_collider p G Z.
\end{minted}
We can now define the central notion of $d$-separation:
\begin{minted}{coq}
Definition d_separated_bool (u v: node) (G: graph) (Z: nodes): bool :=
    forallb (path_is_blocked_bool G Z)
            (find_all_paths_from_start_to_end u v G).
\end{minted}
This definition checks that all paths between $u$ and $v$ are blocked by $Z$, which corresponds to the classic notion of $d$-separation in causal DAGs.

We also define $d$-connectedness in terms of path structure as stated in \autoref{def:d_conn}:
\begin{minted}{coq}
Definition d_connected (p: path) (G: graph) (Z: nodes): Prop :=
    overlap Z (find_mediators_in_path p G) = false /\
    overlap Z (find_confounders_in_path p G) = false /\
    all_colliders_have_descendants_conditioned_on
        (find_colliders_in_path p G) G Z = true.
\end{minted}
We formalize the relationship between $d$-separation and $d$-connectedness in the following theorem:
\begin{minted}{coq}
Theorem d_separated_vs_connected: forall (u v: node) (G: graph) (Z: nodes), 
    d_separated_bool u v G Z = false <->
    exists (l: nodes), acyclic_path (u, v, l) /\
                       is_path_in_graph (u, v, l) G = true /\ 
                       d_connected (u, v, l) G Z.
\end{minted}
This result, which relies on applications of De Morgan’s Law and properties of path enumeration, allows us to move fluidly between the two perspectives.

In our formalization, we typically require paths considered in settings related to $d$-connectedness and $d$-separation to be \textit{acyclic}. While the original definitions in causal-inference literature do not explicitly require paths to be acyclic, the acyclicity of the underlying graph implicitly suggests that the paths of interest are simple. In practice, cyclic paths introduce ambiguity: a node that appears multiple times may play conflicting roles, such as being both a confounder and a collider along the same path. Such overlap creates logical inconsistencies and undermines the clean structural interpretation that $d$-separation provides. Only considering acyclic paths also significantly simplifies reasoning in Rocq, allowing for the assumption that a node’s role in a path is unique and unambiguous. While ensuring paths are acyclic adds technical burden (especially when concatenating paths that overlap), it ultimately makes the formal system more robust and intuitive.

\subsection{Results From Graph Theory and Causal Theory}
In this subsection, we present a collection of formally verified theorems that bridge structural graph theory and causal reasoning. These results arise repeatedly in downstream proofs involving causal semantics. While some of these results may appear intuitive or straightforward, their formal verification in Rocq typically requires careful definitions and auxiliary lemmas.

Often in formalizing causal reasoning, we must construct or merge paths. To ensure correctness, we need criteria for when $d$-connectedness is preserved under concatenation.

\begin{theorem}\label{thm:concat_d_connected}
Let two paths $(u,m,l_1)$ and $(m, v, l_2)$ be $d$-connected, and assume that the two paths do not share any nodes besides $m$. Then, the concatenated path $(u, v, l_1 {++} [m] {++} l_2)$ is $d$-connected if $m$ satisfies one of the following:
\begin{enumerate}
\item It is a mediator in the concatenated path and not in $Z$.
\item It is a confounder in the concatenated path and not in $Z$.
\item It is a collider in the concatenated path and has a descendant in $Z$.
\end{enumerate}
\end{theorem}
\begin{proof}
Assume, for contradiction, that the concatenated path is not $d$-connected. Then there exists some node that blocks the path: either a mediator or confounder in $Z$, or a collider with no descendant in $Z$. This node must lie in the first path, the second path, or at the midpoint $m$. In the first two cases, we contradict the $d$-connectedness of the original paths. In the third case, we contradict the assumption that $m$ satisfies one of the three listed conditions, since by assumption the resulting path is acyclic, so $m$ can only be one of a mediator, confounder, or collider. Hence, the concatenated path must be $d$-connected.
\end{proof}

When composing paths from intersecting segments, we need to avoid cycles. Identifying the first point of overlap between two node lists is thus crucial for disjointness reasoning.

\begin{theorem}\label{thm:first_elt_in_common}
Given two lists of nodes $l_1$ and $l_2$ that share at least one node, there exists a node $x$ such that:
$$l_1 = l_1' {++} [x] {++} l_1'' \text{ and } l_2 = l_2' {++} [x] {++} l_2'',$$
and there is no overlap between $l_1'$ and $l_2'$.
\end{theorem}
\begin{proof}
We proceed via induction on $l_1$, identifying the first point where a node in $l_1$ appears in $l_2$. Once found, we split each list at that node and verify the disjointness of prefixes.
\end{proof}



In many proofs, we begin with a directed path between two nodes (e.g., from a collider to its descendant) but require that the path be acyclic. Fortunately, we can always find such a path without losing essential structure.

\begin{theorem}\label{thm:directed_path_acyclic}
Let $u\neq v$, and suppose there is a directed path from $u$ to $v$. Then there exists an acyclic directed path from $u$ to $v$ consisting of a sublist of the original nodes.
\end{theorem}
\begin{proof}
We perform strong induction on the length of the path. Whenever a cycle is detected (i.e., a repeated node), we remove it, maintaining the original start and end nodes. We prove that subpaths of directed paths remain directed and, by construction, acyclic.
\end{proof}
This result ensures that path-based properties involving descendants can always be reasoned about within the class of acyclic paths.

Continuing with the ideas of intersection points and directed paths, it is often useful to construct new paths by following one directed path to a shared node and continuing along the reverse of the second directed path. This node will always be a collider.

\begin{theorem}\label{thm:directed_paths_intersect_at_collider}
    Let $(u_1, v_1, l_1)$ and $(u_2, v_2, l_2)$ be two directed paths that intersect. Let $x$ be the intersection point guaranteed by \autoref{thm:first_elt_in_common}. Let $P$ be the path from $u_1$ to $u_2$ constructed from following the first path from $u_1$ to $x$ and continuing along the reverse of the second path to $u_2$. Then, $P$ is acyclic, and $x$ is a collider in $P$.
\end{theorem}
\begin{proof}
    By \autoref{thm:directed_path_acyclic}, we can assume the two directed paths are both acyclic. Note that using the notation of \autoref{thm:first_elt_in_common}, $P = (u_1, u_2, l_1' {++} [x] {++} \verb|rev|(l_2'))$. Since there is no overlap between $l_1$ and $l_2'$, $P$ is acyclic.

    Furthermore, since both original paths are directed, both edges into $x$ in the resulting path are incoming. Thus, we prove that $x$ satisfies the collider definition using \verb|colliders_vs_edges_in_path|, as stated in \autoref{sec:med_con_col_functions}.
\end{proof}

\section{Implementation of Function-Based Semantics}\label{app:semantics_impl}
Here we mix in Rocq syntax from our code artifact to give fully rigorous formalizations of key concepts introduced in \autoref{sec:function_based_formal_semantics}.

We choose to capture \textit{node functions} using \nf{}s. A \nf\ assigns a value to a node based on its unobserved term and the values of its parents, which are passed as a pair in our formalization.
\begin{minted}{coq}
Definition nodefun (X: Type): Type := X * list X -> X.
\end{minted}
We represent a \textit{graph function} as a \gf, which maps each node in the graph to its corresponding \nf: 
\begin{minted}{coq}
Definition graphfun {X: Type}: Type := node -> nodefun X.
\end{minted}
Here, \verb|X| denotes the type of values assigned to the nodes, e.g., \verb|bool|, \verb|nat|, or more complex types like \verb|list nat| or finite enumerations.

To track the values assigned to nodes during evaluation, we define the following types:
\begin{minted}{coq}
Definition assignment {X: Type}: Type := node * X.
Definition assignments {X: Type}: Type := list assignment.
\end{minted}
An \verb|assignments X| object functions like a dictionary, mapping nodes to their values. Internally, it is a list of node-value tuples where a lookup returns the most-recent match.

In practical applications, a causal model will likely have known values for some of its nodes but not all. To evaluate others, we define the function \verb|find_value|, which computes the value of a node $u$ in a given graph $\GG$ with \gf\ $g$, using unobserved-term assignments $U$:
\begin{minted}{coq}
Definition find_value {X: Type} (G: graph) (g: graphfun)
                        (u: node) (U: assignments X): option X
\end{minted}

The implementation of \verb|find_value| proceeds in stages. The core function is:
\begin{minted}{coq}
Definition get_value_of_node {X: Type} (u: node) (G: graph) (g: graphfun)
                            (U A: assignments X): option X
\end{minted}
The function \verb|get_value_of_node| takes in $A$, a set of precomputed values. It searches for the parent values of $u$ within $A$. If any parent is not already assigned, it returns \verb|None|. Otherwise, it evaluates \gf\ $g$ on $u$, then evaluates the resulting \nf\ on the precomputed parent values and the unobserved term assigned to $u$ in $U$.

We then define \verb|get_values|, which computes the topological sort of the graph and iteratively computes the value of each node in order using \verb|get_value_of_node|. The resulting list of assignments accumulates values as they are computed. After going through the entire topological sort, \verb|get_values| returns the computed values of all nodes as a set of assignments. Finally, the \verb|find_value| function invokes \verb|get_values| and looks up the value for the desired node.

In the Rocq formalization, \autoref{thm:equiv} is stated generically over all types \verb|X| from which node values may be drawn, subject to the two constraints described in \autoref{sec:main_theorem}. These requirements are captured using a Rocq type class \verb|EqType|, which packages together the domain \verb|X| and its associated decidable equality function.

\section{Finding Disjoint Descendant Paths}\label{app:disjoint_desc_paths}
In \autoref{sec:fpath_intro}, we define $S_4$, the set of nodes that appear on the descendant path from any collider in the $d$-connecting path from $u$ to $v$. However, while $d$-connectedness guarantees the existence of such descendant paths, it does not prevent problematic \textit{overlaps}, such as:
\begin{itemize}
    \item A collider's descendant path might intersect the path itself (see \autoref{fig:desc_path_intersects_path_bad}).
    \item Descendant paths from two different colliders could intersect (see \autoref{fig:desc_paths_intersects_bad}).
\end{itemize}
\begin{figure}
    \begin{subfigure}[b]{.45\textwidth}
        \centering
        \begin{tikzpicture}[
        node distance=1cm and 1cm,
        every node/.style={minimum size=1.2em, font=\large},
        circ/.style={draw, thick, circle, inner sep=2pt},
        dots/.style={font=\scriptsize}
        ]
        
        \node[blue] (u) at (1,1) {$u$};
        \node[right=.5cm of u, blue] (q) {$q$};
        \node[right=.5cm of q, blue] (r) {$r$};
        \node[right=.5cm of r, blue] (s) {$s$};
        \node[right=.5cm of s, circ, blue] (t) {$t$};
        \node[right=.5cm of t, blue] (v) {$v$};
        \node[below=.5cm of r] (x) {$x$};
        \node[above=.5cm of s, circ] (y) {$y$};
            
        \draw[->, thick, blue] (u) -- (q);
        \draw[->, thick, blue] (r) -- (q);
        \draw[->, thick, blue] (r) -- (s);
        \draw[->, thick, blue] (s) -- (t);
        \draw[->, thick, blue] (v) -- (t);
        \draw[->, thick] (q) -- (x);
        \draw[->, thick] (x) -- (s);
        \draw[->, thick] (s) -- (y);
        \end{tikzpicture}
        \caption{The path highlighted in blue is a $d$-connected path, but the descendant path of collider $q$ overlaps with the path at $s$.}
        \label{fig:desc_path_intersects_path_bad}
    \end{subfigure}
    \hfill
    \begin{subfigure}[b]{.45\textwidth}
        \centering
        \begin{tikzpicture}[
        node distance=1cm and 1cm,
        every node/.style={minimum size=1.2em, font=\large},
        circ/.style={draw, thick, circle, inner sep=2pt},
        dots/.style={font=\scriptsize}
        ]
        
        \node[blue] (u) at (1,1) {$u$};
        \node[right=.5cm of u, blue] (q) {$q$};
        \node[right=.5cm of q] (r) {$r$};
        \node[right=.5cm of r, blue] (s) {$s$};
        \node[right=.5cm of s, circ, blue] (t) {$t$};
        \node[right=.5cm of t, blue] (v) {$v$};
        \node[below=.5cm of r, blue] (x) {$x$};
        \node[above=.5cm of s, circ] (y) {$y$};
            
        \draw[->, thick, blue] (u) -- (q);
        \draw[->, thick] (r) -- (q);
        \draw[->, thick] (r) -- (s);
        \draw[->, thick, blue] (s) -- (t);
        \draw[->, thick, blue] (v) -- (t);
        \draw[->, thick, blue] (q) -- (x);
        \draw[->, thick, blue] (x) -- (s);
        \draw[->, thick] (s) -- (y);
        \end{tikzpicture}
        \caption{The alternate path in blue is still $d$-connected, since $q$ is now a mediator. This path is a clean $d$-connected path.}
        \label{fig:desc_path_intersects_path_fix}
    \end{subfigure}
    
    \caption{Given a $d$-connected path from $u$ to $v$ where a descendant path of a collider intersects the path itself, we can construct an alternate $d$-connected path that satisfies \autoref{def:clean_d_path}.}
\end{figure}
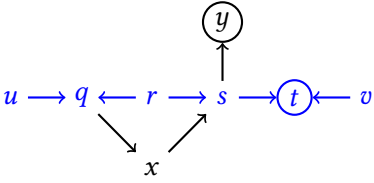
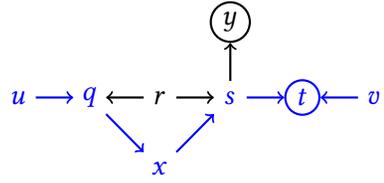

\begin{figure}
    \begin{subfigure}[b]{.45\textwidth}
        \centering
        \begin{tikzpicture}[
        node distance=1cm and 1cm,
        every node/.style={minimum size=1.2em, font=\large},
        circ/.style={draw, thick, circle, inner sep=2pt},
        dots/.style={font=\scriptsize}
        ]
        
        \node[blue] (u) at (1,1) {$u$};
        \node[right=.5cm of u, blue] (q) {$q$};
        \node[right=.5cm of q, blue] (r) {$r$};
        \node[right=.5cm of r, blue] (s) {$s$};
        \node[right=.5cm of s, blue] (t) {$t$};
        \node[right=.5cm of t, blue] (v) {$v$};
        \node[below=.5cm of r] (x) {$x$};
        \node[below=.75cm of s] (y) {$y$};
        \node[below=.5cm of y, circ] (z) {$z$};
            
        \draw[->, thick, blue] (u) -- (q);
        \draw[->, thick, blue] (r) -- (q);
        \draw[->, thick, blue] (r) -- (s);
        \draw[->, thick, blue] (s) -- (t);
        \draw[->, thick, blue] (v) -- (t);
        \draw[->, thick] (q) -- (x);
        \draw[->, thick] (x) -- (y);
        \draw[->, thick] (t) -- (y);
        \draw[->, thick] (y) -- (z);
        \end{tikzpicture}
        \caption{The path highlighted in blue is a $d$-connected path, but the two descendant paths of colliders $q$ and $t$ overlap each other at $y$.}
        \label{fig:desc_paths_intersects_bad}
    \end{subfigure}
    \hfill
    \begin{subfigure}[b]{.45\textwidth}
        \centering
        \begin{tikzpicture}[
        node distance=1cm and 1cm,
        every node/.style={minimum size=1.2em, font=\large},
        circ/.style={draw, thick, circle, inner sep=2pt},
        dots/.style={font=\scriptsize}
        ]
        
        \node[blue] (u) at (1,1) {$u$};
        \node[right=.5cm of u, blue] (q) {$q$};
        \node[right=.5cm of q] (r) {$r$};
        \node[right=.5cm of r] (s) {$s$};
        \node[right=.5cm of s, blue] (t) {$t$};
        \node[right=.5cm of t, blue] (v) {$v$};
        \node[below=.5cm of r, blue] (x) {$x$};
        \node[below=.75cm of s, blue] (y) {$y$};
        \node[below=.5cm of y, circ] (z) {$z$};
            
        \draw[->, thick, blue] (u) -- (q);
        \draw[->, thick] (r) -- (q);
        \draw[->, thick] (r) -- (s);
        \draw[->, thick] (s) -- (t);
        \draw[->, thick, blue] (v) -- (t);
        \draw[->, thick, blue] (q) -- (x);
        \draw[->, thick, blue] (x) -- (y);
        \draw[->, thick, blue] (t) -- (y);
        \draw[->, thick] (y) -- (z);
        \end{tikzpicture}
        \caption{The alternate path in blue is a clean $d$-connected path, since $q$ and $t$ are now mediators, and $y$ has a descendant path to $z$.}
        \label{fig:desc_paths_intersects_fix}
    \end{subfigure}
    
    \caption{Given a $d$-connected path from $u$ to $v$ where the descendant paths of two colliders intersects each other, we can construct an alternate $d$-connected path that satisfies \autoref{def:clean_d_path}.}
\end{figure}
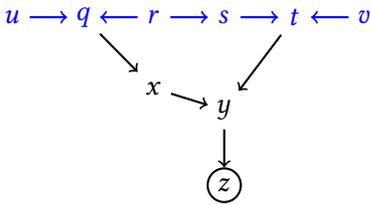
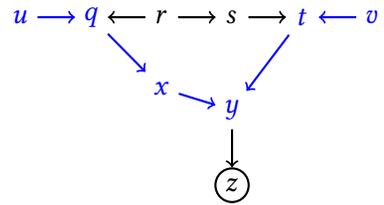
Our construction of $\fpath$ breaks if a node is shared between multiple descendant paths (with multiple predecessors) or if the node is simultaneously a path node (e.g., an $S_2$ node). For instance, in \autoref{fig:desc_path_intersects_path_bad}, node $s$ would be assigned conflicting roles in $\fpath$.

To address this issue, we define a stricter notion of a $d$-connected path:
\begin{definition}\label{def:clean_d_path}
    Let $Z \subseteq \VV$, and let $P$ be a path. Let $c_1, ..., c_k$ be the colliders in $P$. We say $P$ is a \textbf{clean \boldmath $d$-connected path} if $P$ is $d$-connected given $Z$, and for each $c_i$, either $c_i \in Z$, or $c_i \not\in Z$ and there exists a path $Q_i := (c_i, d_n, [d_1, ..., d_{n-1}])$, where $n \geq 1$, satisfying the following conditions:
    \begin{enumerate}
        \item $d_n \in Z$.
        \item $d_i \not\in Z$ for $i < n$.
        \item $d_i \not\in P$ for $i = 1, ..., n$.
        \item For all $c_j \neq c_i$ such that $c_j \not\in Z$, $Q_i$ does not intersect $Q_j$.
    \end{enumerate}
\end{definition}

Our formulation of $\fpath$ requires a clean $d$-connected path. 
Fortunately, when descendant paths intersect either the path or each other, we can construct an alternative clean path using the strategies shown in Figures \ref{fig:desc_path_intersects_path_fix} and \ref{fig:desc_paths_intersects_fix}. These ideas motivate the following result.
\begin{theorem} \label{thm:desc_paths_disjoint}
    If $u$ and $v$ are $d$-connected, then there exists a clean $d$-connected path from $u$ to $v$. 
\end{theorem}

The proof involves extensive case analysis on the structure of the path and descendant relationships. Since the proof is already fully mechanized in Rocq, we present here only a high-level sketch to convey the main ideas.

We introduce the following helper function, which determines the direction of a path immediately after a specific node:
\begin{definition}\label{def:path_out_dir}
    For an acyclic path $P = (u, v, l)$ and a node $w$, let $\dir_P(w) \in \{ \leftarrow, \rightarrow, \bot \}$ represent the direction of the path immediately after the node $w$ in $P$. Concretely, if $w \not\in P$ or $w = v$, then $\dir_P(w) = \bot$. Otherwise, let $w'$ be the node immediately following $w$ on $P$. Then, 
    \[ \dir_P(w) := \begin{cases}
        \leftarrow & \text{if } (w', w) \in \EE \\
        \rightarrow & \text{if } (w, w') \in \EE.
    \end{cases} \]
\end{definition}
Note that since we restrict attention to acyclic paths within acyclic graphs, there is no ambiguity in the $\dir$ function, since each node appears at most once in the path, and for any adjacent pair of nodes $w,w'$, only one of $(w',w)$ or $(w,w')$ can be an edge in the graph. This notion plays a crucial role in the proof; for example, for a directed path $P$, every intermediate node $w$ satisfies $\dir_P(w) =\, \rightarrow$. More generally, in any acyclic path $P$, if $\dir_P(w) =\, \rightarrow$, then $w$ cannot be a collider on $P$.

To prove \autoref{thm:desc_paths_disjoint}, we proceed by induction. At each step, we require the path to satisfy specific structural conditions in order to construct the induction result from the hypothesis. Thus, we establish a stronger, more precisely formulated version of \autoref{thm:desc_paths_disjoint} that incorporates these additional constraints.
\begin{lemma}\label{lem:desc_paths_disjoint_helper}
    Let $P=(u,v,l)$ be a $d$-connected path. Then there exists a clean $d$-connected path $P'=(u,v,l')$ such that for any node $w$ satisfying either of:
    \begin{enumerate}[label=\Roman*.]
        \item $w\in P'$ and $w\not\in P$.
        \item $\dir_P(w) \neq \dir_{P'}(w)$.
    \end{enumerate}
    the following conditions hold:
    \begin{enumerate}
        \item One of:
        \begin{enumerate}
            \item $w\not\in Z$, and there exists a directed path $(w, d, p)$ with $d \in Z$.
            \item $w \in Z$, and $\dir_{P'}(w) \in \{ \leftarrow, \bot \}$.
        \end{enumerate}
        \item If $\dir_{P'}(w) = \, \rightarrow$, then $w\not\in Z$.
    \end{enumerate}
\end{lemma}
\begin{proof}[Proof sketch.]
    We perform induction on the length of $P$. In particular, we identify the node immediately after $u$ in $P$ and invoke the induction hypothesis to construct a clean $d$-connected path from that node to $v$. We then prepend $u$ and analyze the overlaps that are introduced:
    \begin{itemize}
        \item Node $u$ could overlap the induction path or its descendant paths.
        \item If the addition of $u$ makes the next node a collider, then its new descendant path could intersect the induction path or any of its descendant paths.
    \end{itemize}
    For each of these intersections, we construct alternate paths as in Figures \ref{fig:desc_path_intersects_path_fix} and \ref{fig:desc_paths_intersects_fix}.
    
    Conditions 1 and 2 on the induction path allow us to track how path direction and collider status changed from $P$. We must reason about the orientations of the intersection nodes in the original and new paths to ensure that the new path remains $d$-connected, acyclic, clean, and satisfying Conditions 1 and 2.

    Intuitively, nodes that satisfy Condition I are those that originally appeared on a descendant path of a collider in the original path. For example, node $x$ in Figures \ref{fig:desc_path_intersects_path_fix} and \ref{fig:desc_paths_intersects_fix} falls into this category. These nodes satisfy Conditions 1(a) and 2. 

    In the case of nodes appearing on the second path of two intersecting descendant paths, such as $y$ in \autoref{fig:desc_paths_intersects_fix}, the rerouted path $P'$ enters with direction $\leftarrow$, so Condition 2 does not apply. Depending on whether the node is the conditioned descendant or not, it will satisfy either 1(b) or 1(a), respectively.

    Condition II is satisfied by nodes that were colliders in the original path and for which the rerouting process altered the path structure to avoid overlaps. This set includes two types of colliders: those whose descendant paths intersected the original path and the first of a pair of colliders whose descendant paths intersected each other. In both cases, the collider $w$ satisfies Condition 1(a) through its original descendant path. 
    Additionally, since $\dir_P(w) = \,\leftarrow$ and $\dir_{P'}(w) = \,\rightarrow$,as with node $q$ in  Figures \ref{fig:desc_path_intersects_path_fix} and \ref{fig:desc_paths_intersects_fix}, we check Condition 2, which is satisfied because $w\not\in Z$ (since it has a descendant path). For the second collider $w'$ in such a pair, the path direction does not change at $w'$, i.e., $\dir_P(w') = \dir_{P'}(w')$, since the rerouted path reuses the same outgoing edge from $w'$.

    At a high level, by preserving the roles and directions of nodes carefully and redirecting overlaps only when safe, we can maintain the required structure and ensure all conditions of cleanliness hold. The mechanized proof in Rocq is available for full detail.
\end{proof}
It is clear that with \autoref{lem:desc_paths_disjoint_helper}, \autoref{thm:desc_paths_disjoint} is true.

We can now carry out our implementation of $\fpath$, since we can assume by \autoref{thm:desc_paths_disjoint} that there exists a clean $d$-connected path between $u$ and $v$.

In Rocq, we must be able to express not just the path portion of a clean $d$-connected path but also the disjoint descendant paths and the collider that each corresponds to, since we must access the nodes in the descendant paths in order to assign specific \nf{}s.

Let $D$ be of type \verb|assignments (nodes * node)|, mapping collider $c_i$ to a tuple \verb|nodes * node|. If $c_i \in Z$, then $D$ at $c_i$ is simply $([], c_i)$. Otherwise, $D$ at $c_i$ is a tuple $(p_i, d_i)$, where $p_i$ represents the (possibly empty) intermediate nodes on the directed path from $p_i$ to $d_i$, and $d_i \in Z$.

We can then formalize \autoref{def:clean_d_path} to be a \verb|Prop|, given such a set $D$ representing the descendant paths of clean $d$-connected path $(u, v, l)$ in graph $\GG$ with conditioning set $Z$:
\begin{minted}[xleftmargin=2em,linenos]{coq}
Definition descendant_paths_disjoint (D: assignments (nodes * node))
                    (u v: node) (l: nodes) (G: graph) (Z: nodes): Prop :=
  forall (c: node), In c (find_colliders_in_path (u, v, l) G)
      -> get_assigned_value D c = Some ([], c) /\ In c Z
         \/
         exists (p: nodes) (d: node), get_assigned_value D c = Some (p, d)
           /\ In d Z /\ is_directed_path_in_graph (c, d, p) G = true
                     /\ acyclic_path_2 (c, d, p)
           /\ overlap (c :: p) Z = false
           /\ overlap (p ++ [d]) (u :: l ++ [v]) = false
           /\ forall (c' d': node) (p': nodes),
              c =? c' = false /\ get_assigned_value D c' = Some (p', d')
              -> overlap (c :: p ++ [d]) (c' :: p' ++ [d']) = false.
\end{minted}
In the above, line 4 is the case that the collider $c$ is conditioned on, and thus we do not need a descendant path. Lines 6-8 describe the alternate case, that instead there is an acyclic, directed path from $c$ to a conditioned descendant $d$ along $p$. Line 9 ensures that $d$ is the first node in the path that is in $Z$, and line 10 ensures that the descendant path does not intersect the path from $u$ to $v$ itself. Lines 11-13 ensure that no two different descendant paths given by $D$ overlap each other. 

\begin{definition}\label{def:descendant_map}
    Given a clean $d$-connected path $P = (u, v, l)$ conditioned on $Z$, then $D$ is a \textbf{descendant map} for $P$ if $\verb|descendant_paths_disjoint|(D, u, v, l, \GG, Z)$ holds.
\end{definition}

It is clear by \autoref{def:clean_d_path} that any clean $d$-connected path will have at least one descendant map.

\section{Existence of Node Set Assignments}\label{app:exset}
Here, we provide the Rocq definitions for the \nf{}s described in \autoref{sec:fpath_intro} based on the set into which each node of $\VV$ falls:
\begin{itemize}
    \item $S_1$ (sources). Use unobserved term directly:
\begin{minted}{coq}
Definition f_unobs (X: Type) (val: X * list X): X := fst val.
\end{minted}
    \item $S_2$ (transmitters). Use value of parent on $P$, where $i$ is the index of the designated parent node in the transmitter's parent list:
\begin{minted}{coq}
Definition f_parent_i (X: Type) (i: nat) (val: X * list X): X :=
    nth_default (fst val) (snd val) i.
\end{minted}
    \item $S_3$ (colliders). Enforce equality of two parents on $P$, where $x = A_Z(d)$, $d$ is the descendant of the collider in $Z$ (possibly the collider itself), $y \neq x$, and $i$ and $j$ are the indices of the two parents on $P$ in the collider's parent list:
\begin{minted}{coq}
Definition f_equate_ij (X: Type) `{EqType X} (i j: nat) (x y: X)
                       (val: X * (list X)): X :=
    if eqb (nth_default (fst val) (snd val) i)
           (nth_default (fst val) (snd val) j) then x else y.
\end{minted}
    \item $S_4$ (descendants). Use \verb|f_parent_i X i| as in $S_2$, where the parent in the descendant path is at index $i$ in the descendant's parent list.
    \item $S_5$ ($Z$-residual). Force to assigned value in $A_Z$:
\begin{minted}{coq}
Definition f_constant (X: Type) (res: X) (val: X * (list X)): X := res.
\end{minted}
    \item $S_6$ (residual). Assign any arbitrary \nf.
\end{itemize}

We define the following assignments $A_2, A_3, A_4$, corresponding to the nodes of $S_2, S_3$, and $S_4$ respectively:
\begin{itemize}
    \item $A_2$: \verb|assignments nat|. Maps node $w \in S_2$ to $i$, corresponding to the index of $w$'s single parent in the path in $w$'s parent list.
    \item $A_3$: \verb|assignments (nat * nat * X * X)|. Maps node $w \in S_3$ to $(i, j, x, y)$, where $i$ and $j$ are the indices of $w$'s two parents in the path in $w$'s parent list, $x$ is the value of $w$'s conditioned descendant given by $A_Z$, and $y$ is some value unequal to $x$.
    \item $A_4$: \verb|assignments nat|. Maps node $w \in S_4$ to $i$, corresponding to the index of the previous node in the descendant path in $w$'s parent list. Here we use any non-overlapping descendant path described above.
\end{itemize}

For example, if for the path $p_w \rightarrow w \rightarrow \cdots$, $A_2(w) = i$, then \verb|f_parent_i X i| returns $f(p_w)$. We show how to prove that such assignments exist, as well as how to assemble them into a unified graph function $\fpath$.

For the clean $d$-connected path, as described in \autoref{app:disjoint_desc_paths}, from $u$ to $v$, given the set $S_1$ and the assignments $A_2$, $A_3$, $A_4$, as well as some arbitrary default \nf\ $h_\textsf{def}$, and the conditioned assignments $A_Z$, we can implement the \gf\ $\fpath$:
\begin{equation} \label{eq:fpath}
    \fpath := \verb|g_path|(S_1, A_2, A_3, A_4, A_Z, h_\text{def}),
\end{equation}
where
\begin{minted}{coq}
Definition g_path (X: Type) `{EqType X} (S1: nodes) (A2: assignments nat)
                  (A3: assignments (nat * nat * X * X)) (A4: assignments nat)
                  (AZ: assignments X) (default: nodefun X)
                  (w: node): nodefun X :=
  match member S1 w with
  | true => f_unobs X
  | false =>
          match get_assigned_value A2 w with
          | Some i => f_parent_i X i
          | None =>
                 match get_assigned_value A3 w with
                 | Some (i, j, x, y) => f_equate_ij X i j x y
                 | None =>
                        match get_assigned_value A4 w with
                        | Some i => f_parent_i X i
                        | None =>
                               match get_assigned_value AZ w with
                               | Some x => f_constant X x
                               | None => default
                               end
                        end
                 end
          end
  end.
\end{minted}
We now must show the existence of these sets and assignments for arbitrary $d$-connected nodes $u$ and $v$. For the following lemmas, let $P$ be a clean $d$-connected path conditioned on $Z$ between $u$ and $v$ in $\GG$. Let $S_1, S_2, S_3, S_4, S_5$, and $S_6$ be the sets described by the partition for $\fpath$ given $P$. Let $D$ be any descendant map for $P$, as described in \autoref{def:descendant_map}.

\begin{lemma}\label{lem:A2_existence}
     There exists $A_2$, assignments for the nodes of $S_2$ to $\mathbb{N}$, satisfying that for all $(w, i) \in A_2$, there exists a node $a$ such that $a$ is the $i$-th node in $\textsf{Pa}(w)$, and either $a \rightarrow w$ or $w \leftarrow a$ is a subpath in $P$.
\end{lemma}
\begin{proof}
    We perform strong induction on the length of $P$. In the base case, $P$ is simply an edge between $u$ and $v$. Suppose $P$ is $u \rightarrow v$. Then, $S_2 = \{ v \}$. Since $u$ is a parent of $v$, we can determine $i_u$ such that $u$ is the $i_u$-th node in $\textsf{Pa}(v)$. Let $A_2 := \{ v : i_u \}$. Similarly, if $P$ is $u \leftarrow v$, then $S_2 = \{u\}$. Let $i_v$ be the index of $v$ in $\textsf{Pa}(u)$. Then, let $A_2 := \{ u : i_v \}$. 

    For the induction step, suppose the statement is true for paths of length $k$, where $2 \leq k < n$. Suppose $P$ has length $n > 2$. Let $w$ be the node following $u$ in $P$. We consider different cases for the edge orientations at the start of the path surrounding $w$:

    \textbf{\boldmath Case 1: $u \leftarrow w  \cdots v$}. By the induction hypothesis, there is a set $A_2'$ corresponding to $S_2'$ for the path $w \cdots v$. Note that if the path continues into $w$, e.g., $w \leftarrow \cdots v$, then $w \in S_2'$, and $w \in S_2$. If the path continues out of $w$, e.g., $w \rightarrow \cdots v$, then $w\in S_1'$, and $w \in S_1$. Thus, $S_2 = \{ u \} \cup S_2'$. Let $i_w$ be the index of $w$ in $\pa(u)$. Then, we define $A_2 := \{ u: i_w\} \cup A_2'$.

    \textbf{\boldmath Case 2: $u \rightarrow w \rightarrow \cdots v$}. By the induction hypothesis, there is a set $A_2'$ corresponding to $S_2'$ for the path $w \rightarrow \cdots v$. Note that $w \in S_1'$, but $w \in S_2$. Furthermore, $u \in S_1$. Thus, $S_2 = \{ w \} \cup S_2'$. Let $i_u$ be the index of $u$ in $\pa(w)$. Then, we define $A_2 := \{ w: i_u \} \cup A_2'$.

    \textbf{\boldmath Case 3a: $u \rightarrow w \leftarrow v$}. Consider the case of a three-node path, in which $w$ is a collider. Then, $S_2 = \emptyset$, so let $A_2 := \emptyset$.
    
    \textbf{\boldmath Case 3b: $u \rightarrow w \leftarrow w' \cdots v$}. Now suppose there are more than three nodes in the path. By the induction hypothesis, there is a set $A_2'$ corresponding to $S_2'$ for the path $w' \cdots v$. Note that $w \in S_3$, and $u \in S_1$. Thus, $S_2 = S_2'$. So, we define $A_2 := A_2'$.

    Thus, we can define $A_2$ as desired for a path of length $n$ for all cases of edge orientations.
\end{proof}

\begin{lemma}\label{lem:A3_existence}
    There exists $A_3$, assignments for the nodes of $S_3$ to $\mathbb{N} \times \mathbb{N} \times \verb|X| \times \verb|X|$, satisfying that for all $(w, (i, j, x, y)) \in A_3$, there exist nodes $a, b$, such that $a$ and $b$ are the $i$-th and $j$-th nodes in $\pa(w)$, respectively, and $a \rightarrow w \leftarrow b$ is a subpath in $P$. Furthermore, there exist list of nodes $p$ and node $d$ such that $(w, (p, d)) \in D$, $A_Z(d) = x$, and $x \neq y$.
\end{lemma}
\begin{proof}
    Again, we proceed via strong induction on the length of $P$. For the base case, $P$ has only two endpoints and no intermediate nodes, so $S_3 = \emptyset$. Thus, $A_3 := \emptyset$.

    For the induction step, we follow the same setup as \autoref{lem:A2_existence}: suppose the statement is true for paths of length $k$, where $2 \leq k < n$ and that $P$ has length $n$. Let $w$ be the node following $u$ in $P$. By the induction hypothesis, there is a set $A_3'$ corresponding to $S_3'$ for the path $w \cdots v$. Suppose $w$ is not a collider in $P$. Then, $S_3 = S_3'$, and $D$ is also a descendant map for the path $w \cdots v$. Thus, we define $A_3 := A_3'$.

    Now suppose $w$ is a collider in $P$. We consider two cases:

    \textbf{Case 1: $u \rightarrow w \leftarrow v$}. Here, $P$ is the three-node path in which $w$ is a collider. Then, $S_3 = \{ w \}$. Since $u$ and $v$ are parents of $w$, they must appear in $\pa(w)$. Let $i_u$ and $i_v$ be the indices of $u$ and $v$ in $\pa(w)$, respectively. If $w\in Z$, let $x := A_Z(w)$. Otherwise, there exists $p,d$ such that $(w, (p, d)) \in D$. Let $x := A_Z(d)$. Since we assume that \verb|X| has at least two distinct elements, we can choose a $y$ such that $y \neq x$. Let $A_3 := \{ w: (i_u, i_v, x, y) \}$.

    \textbf{Case 2: $u \rightarrow w \leftarrow w' \cdots v$}. Now there are more than three nodes in the path. Consider the path $w' \cdots v$. It is clear that $D$ is also a descendant map for this path, since any overlaps in the descendant paths with each other or with the path itself would also occur for $P$. Then, by the induction hypothesis, there exists $A_3'$ corresponding to $S_3'$ for the path $w' \cdots v$. Note that $S_3 = \{ w \} \cup S_3'$. Thus, we only need to add on an entry for $w$ to $A_3'$, and we do so using the same methods as Case 1, defining $A_3 := \{ w: (i_u, i_v, x, y)\} \cup A_3'$.
\end{proof}

\begin{lemma}\label{lem:A4_existence}
    There exists $A_4$, assignments for the nodes of $S_4$ to $\mathbb{N}$, satisfying that for all $(w, i_a) \in A_4$, there exists $i_D$, nodes $c, d, a$ and list of nodes $p$, such that $(c, (p, d)) \in D$ where $c \neq d$, $w$ is at index $i_D$ in path $(c, d, p)$, $a$ is at index $i_D - 1$ in $(c, d, p)$, and $a$ is the $i_a$-th node in $\pa(w)$.
\end{lemma}
\begin{proof}
    If $w \in S_4$, then $w$ must belong to the descendant path of some collider in $D$. We define the following function which, given a descendant path $(c, d, p)$, assigns the correct index bindings for the nodes on that path.

\begin{minted}{coq}
Fixpoint get_A4_nodes_for_path (c d: node) (p: nodes) (G: graph):
                               option (assignments nat) :=
  match p with
   | [] => match index (find_parents d G) c with
           | Some i => Some [(d, i)]
           | None => None
           end
   | h :: t => match index (find_parents h G) c with
               | Some i => match get_A4_nodes_for_path h d t G with
                           | Some r => Some ((h, i) :: r)
                           | None => None
                           end
               | None => None
               end
   end.
\end{minted}
Induction on the length of $p$ will show that $\verb|get_A4_nodes_for_path|(c, d, p, \GG)$ exists as long as $(c, d, p)$ is a directed path. Then, given $D$, we can go through each node of $S_3$ (colliders) and append the results of \verb|get_A4_nodes_for_path| together to get $A_4$. Induction on the length of the path (which affects the number of nodes in $S_3$) shows that this result exists as long as $D$ is a descendant map.

By construction, it is easy to see that any $(w, i_a) \in A_4$ computed as described above must correspond to a node in a descendant path and the index of the previous node in the path in $\pa(w)$. Mechanistically, induction on the colliders of $P$ will show that this result is true.
\end{proof}

Thus, we have shown the existence of $A_2, A_3$, and $A_4$ for arbitrary clean $d$-connected path $P$ from $u$ to $v$. We can then use \autoref{eq:fpath} to define the \gf\ $\fpath$.

\section{Equating Node Values and Conditioning on \texorpdfstring{$Z$}{Z}} \label{app:equate_node_values}
We will now show that for any choice of source-fixed unobserved-terms assignments $U$, $\fpath$ indeed equates noncollider values and properly conditions on $Z$. 

Let $P$ be a clean $d$-connected path given $Z$ and $A_Z$ from $u$ to $v$ in $\GG$. Let $S_1, ..., S_6$ be the partition described in \autoref{sec:fpath_intro}. Let $U$ be any unobserved-terms assignments source-fixed to $\alpha$, as defined in \autoref{def:source-fixed}. Let $A_2, A_3, A_4$, and $\fpath$ be as described in \autoref{app:exset}.

\begin{theorem}\label{thm:equate_values}
    For all noncollider nodes $w \in P$, $\fpath_U(w) = \alpha$.
\end{theorem}
\begin{proof}
    Note that all noncollider nodes $w \in P$ must be in $S_1$ or $S_2$. It is clear that for all $w \in S_1$, $\fpath_U(w) = \alpha$, since $\fpath_{U}(w) = U(w) = \alpha$ (recall that the \nf\ corresponding to a source is \verb|f_unobs|).

    Now consider any transmitter $w\in S_2$. Let $i$ be the mapping of $w$ given by $A_2$. Let $a$ be the $i$-th node in $\pa(w)$. Then, either $a \rightarrow w$ or $w \leftarrow a$ is a subpath of $P$. Note that since the \nf\ corresponding to $w$ is \verb|f_parent_i|, $\fpath_U(w) = \fpath_U(\pa(w)_i) = \fpath_U(a)$.

    Consider first the case that $a \rightarrow w$ is a subpath of $P$. Let the index of $w$ in $P$ be $i_w$. We perform induction on $i_w$: in the base case, $i_w = 1$ with zero-based indexing ($a$ is the first node, and $w$ is the second node). Then, $a$ must be in $S_1$, so we know that $\fpath_U(a) = \alpha$. Thus, $\fpath_U(w) = \fpath_U(a) = \alpha$, as desired.
    
    For the induction step, assume that $i_w > 1$. Our induction hypothesis is the following: for any two nodes $a'$ and $w'$ such that $a' \rightarrow w'$ is a subpath of $P$ and the index of $w'$ in $P$ is $i_w - 1$, we must have $\fpath_U(a') = \alpha$. Note that the index of $a$ in $P$ is exactly $i_w-1 \geq 1$, so there must be a node preceding $a$ in $P$: call it $a'$. If the edge orientation is $a' \rightarrow a$, then we can apply the induction hypothesis to get that $\fpath_U(a') = \alpha$, and thus $\fpath_U(a) = \alpha$. If the edge orientation is $a' \leftarrow a$, then $a \in S_1$, so we know that $\fpath_U(a) = \alpha$. In all cases, $\fpath_U(a) = \alpha$, so $\fpath_U(w) = \alpha$.

    The case of $w \leftarrow a$ being the subpath of $P$ is very similar, except that we would induct on the index of $w$ in the \textit{reverse} of $P$; in other words, the base case would be when $a$ and $w$ are the first and second nodes, respectively, in the reverse path; thus $w$ and $a$ are the second-to-last and last nodes, respectively, in $P$. This allows for almost identical logic to the above case.
\end{proof}
We have now shown that for all noncollider nodes $w_i, w_j \in P$, $\fpath_U(w_i) = \fpath_U(w_j)$. Importantly, this result means that $\fpath_U(u) = \fpath_U(v)$, which intuitively tells us that $u$ and $v$ cannot be semantically separated, as their values are equal.

Of course, $\fpath$ is of no use if it does not properly condition on $Z$ given assignments $A_Z$. Thus, we must show that all nodes in $Z$ evaluate to the correct values. Recall that $U$ is source-fixed to $\alpha$. 
\begin{theorem}\label{thm:condition_on_Z}
    For all nodes $z\in Z$, $\fpath_U(z) = A_Z(z)$.
\end{theorem}
\begin{proof}
    First, we show that for all nodes $w\in S_3$ (colliders), if $w$ maps to $(i, j, x, y)$ in $A_3$, then $\fpath_U(w) = x$. Let $a$ and $b$ be the $i$-th and $j$-th nodes of $\pa(w)$, respectively. Then, $a \rightarrow w \leftarrow b$ is a subpath of $P$. Since the \nf\ corresponding to $w$ is \verb|f_equate_ij|, \[\fpath_U(w) := \begin{cases}
    x & \text{if } \fpath_U(a) = \fpath_U(b) \\
    y & \text{otherwise}.
    \end{cases}\]
    Since $a$ and $b$ both have arrows out towards $w$, neither can be in $S_3$. Thus, both are in $S_1\cup S_2$ (noncollider nodes in $P$). Then, by \autoref{thm:equate_values}, $\fpath_U(a) = \fpath_U(b)$. Thus, $\fpath_U(w) = x$.

    Now, consider any $z\in Z$. By construction of the partition, $z \in S_3 \cup S_4 \cup S_5$. Suppose $z\in S_5$ (one of the residual nodes in $Z$). Then, by definition, $\fpath_U(z) = A_Z(z)$. Furthermore, if $z \in S_3$ is a collider, then $z$ does not need a descendant path, so $A_Z(z) = x$ by construction of $A_3$.
    Then, by the above paragraph, $\fpath_U(z) = x = A_Z(z)$, as desired.

    It remains to show that $\fpath_U(z) = A_Z(z)$ if $z\in S_4$ is the conditioned descendant of a collider.
    Let $c$ be the corresponding collider, and let $[c] {++} p {++} [z]$ be the path from collider to descendant
    (more concretely, $(c, (p, z)) \in D$, where $D$ is a descendant map, as described in \autoref{def:descendant_map}).
    We will show that for all nodes $w \in S_4$ in the descendant path, $\fpath_U(w) = \fpath_U(c)$. Let $i$ be the mapping of $w$ given by $A_4$. Let $a$ be the $i$-th node in $\pa(w)$. Since the \nf\ corresponding to $w$ is \verb|f_parent_i|, $\fpath_U(w) = \fpath_U(\pa(w)_i) = \fpath_U(a)$.
    
    We proceed via induction on the index of $w$ in the path. For the base case, we assume $w$ is the first node of $p$. Then, $a = c$, so $\fpath_U(w) = \fpath_U(c)$. For the induction hypothesis, we assume all nodes in the descendant path prior to $w$ evaluate to $\fpath_U(c)$. We can then apply the induction hypothesis on $a$, and we thus have $\fpath_U(w) = \fpath_U(a) = \fpath_U(c)$, as desired.

    Then, since $z$ is a node in the descendant path $[c] {++} p {++} [z]$, we have $\fpath_U(z) = \fpath_U(c)$. We furthermore know, by the first paragraph of the proof, that $\fpath_U(c) = x$, where $(i_c, j_c, x, y)$ is the mapping of $c$ given by $A_3$. Since $z$ is the corresponding descendant, $x = A_Z(z)$ by construction of $A_3$, so $\fpath_U(c) = A_Z(z)$. Thus, $\fpath_U(z) = A_Z(z)$, as desired.
\end{proof}

\section{Satisfying the Conditions of Semantic Separation}\label{app:forward}
To prove that the unobserved-terms assignments given in \autoref{eq:sequence_U} satisfy the conditions of \autoref{def:ci}, we must establish relationships between the nodes of $S_1$ with other nodes in the path, particularly nodes of which they are unblocked ancestors.

\begin{lemma}\label{lem:first_S1_node_anc}
    The first node in $P$ that is a member of $S_1$ is an unblocked ancestor of $u$.
\end{lemma}
\begin{proof}
    Proceed via induction on the length of $P$. For the base case, $P$ consists of only $u$ and $v$. If $P$ is $u \rightarrow v$, then $u$ is the first member of $S_1$, and $u \in \anc{u}$. If $P$ is $u \leftarrow v$, then $v$ is the first member of $S_1$, and it is clear that $v \in \anc{u}$, since $v \not\in Z$ by assumption.

    For the induction hypothesis, we assume that for paths of length $n$, the first node in $S_1$ is an unblocked ancestor of the first node in the path. Suppose $P$ has length $n + 1$, where $n \geq 2$. If the edge out of $u$ is outwards, then $u$ is the first member of $S_1$, and $u \in \anc{u}$. Otherwise, $P$ is $u \leftarrow w \cdots v$ for some node $w$. Let $P'$ be the subpath $w \cdots v$ with respective partitions $S_1'$, $S_2'$, etc. Note that $u \in S_2$, so the first node in $P'$ in $S_1'$, $a$, will also be the first node in $S_1$. By the induction hypothesis, $a$ is an unblocked ancestor of $w$. Note that since $P$ is $d$-connected, $w \not\in Z$, since $w$ is not a collider. Thus, $a \in \anc{u}$.
\end{proof}

\begin{lemma}\label{lem:consecutive_S1_nodes}
    If $x,y \in S_1$ are consecutive sources (no nodes between them on $P$ are in $S_1$), then there exists a node $z\in Z$ such that $x \in \anc{z}$ and $y \in \anc{z}$.
\end{lemma}
\begin{proof}
    Intuitively, $x$ and $y$ must both have arrows out; in other words, $x \rightarrow \cdots \leftarrow y$ is a subpath of $P$. Then, at some point in the subpath, there must be a collider. Since $P$ is $d$-connected, this collider has a conditioned descendant $z \in Z$, of which $x$ and $y$ are unblocked ancestors.

    To formally prove the statement, we induct on the length of the path. Note that $P$ must have length at least 3 in order to have two nodes in $S_1$. Thus, for the base case, $P$ is $u \rightarrow w \leftarrow v$ (the other arrow orientations result in only one node in $S_1$), where $u = x$ and $v = y$. If $w\in Z$, then let $z = w$. It is clear that $x \in \anc{w}$ and $y \in \anc{w}$. Otherwise, $w$ must have a conditioned descendant $z \in Z$ such that there is a directed path from $w$ to $z$ that does not go through other nodes in $Z$. Then, it is clear that $x \in \anc{z}$ and $y \in \anc{z}$.

    Suppose the statement is true for paths $P'$ of length $n$ and corresponding set of nodes $S_1'$. Suppose $P$ has length $n+1$, where $n \geq 3$, so $P$ is $u \leftrightarrow w_1 \leftrightarrow w_2 \cdots v$. Let $P'$ be the subpath $w_1 \cdots v$. If $P$ has an arrow into $u$, then $u \in S_2$, so $x,y$ are still consecutive nodes in $S_1'$ corresponding to $P'$. Thus, we can apply the induction hypothesis on $P'$ to get the desired $z$.

    Now consider the case that $P$ has an arrow out of $u$. If $u \neq x$, then $x$ and $y$ are again still consecutive nodes in $S_1'$ corresponding to $P'$, so we again can apply the induction hypothesis on $P'$. Suppose $u = x$. We first consider the case that $P$ is $u \rightarrow w_1 \rightarrow w_2 \cdots v$. Then, note that the first node of $S_1'$ will be $w_1$, and the second will be $y$. Thus, $w_1$ and $y$ are consecutive nodes in $S_1'$, so we can find a $z\in Z$ such that $w_1 \in \anc{z}$ and $v\in \anc{z}$. Then, $u \in \anc{z}$ since $u \rightarrow w_1$ is an edge and $w_1 \not\in Z$, since $w_1$ is a mediator.

    We now consider the case that $P$ is $u \rightarrow w_1 \leftarrow w_2 \cdots v$. By \autoref{lem:first_S1_node_anc}, $y \in \anc{w_1}$. Since $w_1$ is a collider, it has a conditioned descendant $z\in Z$. Applying the same logic as the base case, we have that $u = x \in \anc{z}$, as desired.
\end{proof}

Letting $U_0, ... , U_\ell$ be the sequence given in \autoref{eq:sequence_U}, note that $U_0$ satisfies Condition 1 of \autoref{def:ci}. Furthermore, $\fpath_{U_0}(v) = \alpha \neq \beta = \fpath_{U_\ell}(v)$ by \autoref{thm:equate_values}. Thus, if we show that the sequence satisfies the remaining conditions of \autoref{def:ci}, then we will show that $u$ and $v$ are not semantically separated.

\begin{theorem}\label{thm:sequence_U}
    For the sequence defined in \autoref{eq:sequence_U}, each two consecutive unobserved-terms assignments $U_{i-1}, U_i$ in the sequence differ from each only for
    $$a' \in \bigcup_{\substack{z\in Z \\ \exists a \in \anc{z}, \\ U_{i-2}(a) \neq U_{i-1}(a)}} \anc{z}.$$
\end{theorem}
\begin{proof}
    Note that the statement applies only to paths with at least two sources, since $\ell$ must be at least 2 for $U_{i-2}$ to make sense. Consider $U_i$ and $U_{i-1}$ for some $i \geq 2$. By definition, they differ from each other only for $s_i$. By \autoref{lem:consecutive_S1_nodes}, there is a node $z \in Z$ such that $s_{i-1}, s_i \in \anc{z}$. Note that by definition, $U_{i-2}(s_{i-1}) = \alpha \neq \beta = U_{i-1}(s_{i-1})$. Thus, the statement holds.

    In order to prove the above formally, we proceed via induction on the length of $S_1$, proving a key lemma that $S_1'$ corresponding to the subpath $s_i \cdots v$ of $P$ is the subset of $S_1$ containing nodes beginning from $s_i$ and continuing to the end of the subpath.
\end{proof}

We now fill in the remaining details of the forward direction of \autoref{thm:equiv}, restated below for convenience.
\equiv*
\begin{proof}[Proof (Forward).]
    We aim to show that if $u$ and $v$ are semantically separated given $Z$, then they are $d$-separated given $Z$ in $\GG$.
    
    We proceed via the contrapositive. Suppose that $u$ and $v$ are not $d$-separated, so there exists a clean $d$-connected path $P$ from $u$ to $v$. Let $U_0, U_1, ..., U_\ell$ be defined as in \autoref{eq:sequence_U}. We show that all conditions of \autoref{def:ci} are satisfied, using graph function $\fpath$:
    \begin{enumerate}
        \item Since $U_0$ is source-fixed to $\alpha$, $\fpath_{U_0}(u) = \alpha$, and the function properly conditions on $Z$ by Theorems \ref{thm:equate_values} and \ref{thm:condition_on_Z}, respectively.
        \item By construction, $U_1$ differs from $U_0$ for only $s_1$, where $s_1 \in \anc{u}$ by \autoref{lem:first_S1_node_anc}. To show that $\fpath_{U_1}(u) = \beta$, we consider the edge orientation associated with $u$ on $P$: if $P$ has an edge out of $u$, then $u\in S_1$, so $\fpath_{U_1}(u) = U_1(u) = \beta$. Otherwise, $u\in S_2$. Via induction on the length of the path, we show that the chain of $S_2$ nodes starting at $u$ will terminate at $s_1$, and all nodes along the chain will take on the value of $s_1$, which is $\beta$.
        \item By \autoref{thm:sequence_U}, this condition is satisfied.
        \item Note that $\ell = |S_1|$, and $|S_1| \leq |\VV|$, so the condition is satisfied.
        
        \item Since $U_\ell$ is source-fixed to $\beta$, $\fpath_{U_\ell}(u) = \beta$, and the function properly conditions on $Z$ by Theorems \ref{thm:equate_values} and \ref{thm:condition_on_Z}, respectively.
    \end{enumerate}
    By \autoref{thm:equate_values}, $\fpath_{U_0}(v) = \alpha \neq \beta = \fpath_{U_\ell}(v)$. Thus, $u$ and $v$ are not semantically separated.
\end{proof}

\section{The Backward Direction for Short Sequences}\label{app:backward}

\valueaffectedbyunblocked*
\begin{proof}
    We proceed via strong induction on the index of $w$ in the topological sort of $\GG$. For the base case, suppose $w$ is the first node in the topological sort. Then, it has no parents, so $f_U(w)$ and $f_{U'}(w)$ depend only on $U(w)$ and $U'(w)$, respectively. Thus, it must be true that $U(w) \neq U'(w)$, so we can simply let $a = w$, where clearly $w \in \anc{w}$.

    For the induction hypothesis, assume that the statement is true for all nodes with index at most $i$ in the topological sort of $\GG$, and suppose the index of $w$ is $i + 1$. If $U(w) \neq U'(w)$, we can once again simply let $a = w$. Assume $U(w) = U'(w)$. Then, there must be some $w' \in \pa(w)$ such that $f_U(w') \neq f_{U'}(w')$, since $f_U(w) \neq f_{U'}(w)$. Since $w' \in \pa(w)$, the index of $w'$ in the topological sort of $\GG$ must be less than $i + 1$. Thus, by the induction hypothesis, there exists a node $a \in \anc{w'}$ such that $U(a) \neq U'(a)$. Note that since $f_U$ and $f_{U'}$ both properly condition on $Z$, it must be true that $w' \not\in Z$, since $f_U(w') \neq f_{U'}(w')$. Thus, $a \in \anc{w}$, as desired.
\end{proof}

\begin{figure}
    \begin{subfigure}[b]{.4\textwidth}
        \centering
        \begin{tikzpicture}[
        node distance=1cm and 1cm,
        every node/.style={minimum size=1.2em, font=\large},
        circ/.style={draw, thick, circle, inner sep=2pt},
        dots/.style={font=\scriptsize}
        ]
        
        \node (w1) at (1,1) {$w_1$};
        \node[below=.5cm of w1, dots] (dots) {$\vdots$};
        \node[below=.5cm of dots] (w2) {$w_2$};
        
        \draw[->, thick] (w1) -- (dots);
        \draw[->, thick] (dots) -- (w2);
        \end{tikzpicture}
        \caption{$a = w_1$ and has a directed path to $w_2$.}
        \label{fig:w1w2directed}
    \end{subfigure}
    \hfill
    \begin{subfigure}[b]{.4\textwidth}
        \centering
        \begin{tikzpicture}[
        node distance=1cm and 1cm,
        every node/.style={minimum size=1.2em, font=\large},
        circ/.style={draw, thick, circle, inner sep=2pt},
        dots/.style={font=\scriptsize}
        ]
        
        \node (w2) at (1,1) {$w_2$};
        \node[below=.5cm of w2, dots] (dots) {$\vdots$};
        \node[below=.5cm of dots] (w1) {$w_1$};
        
        \draw[->, thick] (w2) -- (dots);
        \draw[->, thick] (dots) -- (w1);
        \end{tikzpicture}
            \caption{$a = w_2$ and has a directed path to $w_1$.}
            \label{fig:w2w1directed}
    \end{subfigure}

    \begin{subfigure}[b]{.4\textwidth}
        \centering
        \begin{tikzpicture}[
        node distance=1cm and 1cm,
        every node/.style={minimum size=1.2em, font=\large},
        circ/.style={draw, thick, circle, inner sep=2pt},
        dots/.style={font=\scriptsize}
        ]
        
        \node (a) at (1,1) {$a$};
        \node[below left=.5cm of a, dots] (dots1) {\reflectbox{$\ddots$}};
        \node[below right=.5cm of a, dots] (dots2) {$\ddots$};
        \node[below left=.5cm of dots1] (w1) {$w_1$};
        \node[below right=.5cm of dots2] (w2) {$w_2$};
        
        \draw[->, thick] (a) -- (dots1);
        \draw[->, thick] (dots1) -- (w1);
        \draw[->, thick] (a) -- (dots2);
        \draw[->, thick] (dots2) -- (w2);
        \end{tikzpicture}
        \caption{The two directed paths from $a$ to $w_1$ and $w_2$ do not overlap.}
        \label{fig:confounderpatha}
    \end{subfigure}
    \hfill
    \begin{subfigure}[b]{.4\textwidth}
        \centering
        \begin{tikzpicture}[
        node distance=1cm and 1cm,
        every node/.style={minimum size=1.2em, font=\large},
        circ/.style={draw, thick, circle, inner sep=2pt},
        dots/.style={font=\scriptsize}
        ]
        
        \node (a) at (1,1) {$a$};
        \node[below=.5cm of a, dots] (dots1) {$\vdots$};
        \node[below=.5cm of dots1] (a2) {$a'$};
        \node[below left=.5cm of a2, dots] (dots3) {\reflectbox{$\ddots$}};
        \node[below right=.5cm of a2, dots] (dots4) {$\ddots$};
        \node[below left=.5cm of dots3] (w1) {$w_1$};
        \node[below right=.5cm of dots4] (w2) {$w_2$};
        
        \draw[->, thick] (a) -- (dots1);
        \draw[->, thick] (dots1) -- (a2);
        \draw[->, thick] (a2) -- (dots3);
        \draw[->, thick] (a2) -- (dots4);
        \draw[->, thick] (dots3) -- (w1);
        \draw[->, thick] (dots4) -- (w2);
        \end{tikzpicture}
        \caption{The two directed paths from $a$ to $w_1$ and $w_2$ overlap at $a'$.}
        \label{fig:confounderpathb}
    \end{subfigure}
    
    \caption{Given that $a \in \anc{w_1} \cap \anc{w_2}$, there are four possible cases for the path structure between $w_1$ and $w_2$, resulting in either a directed path or a single-confounder path.}
\end{figure}
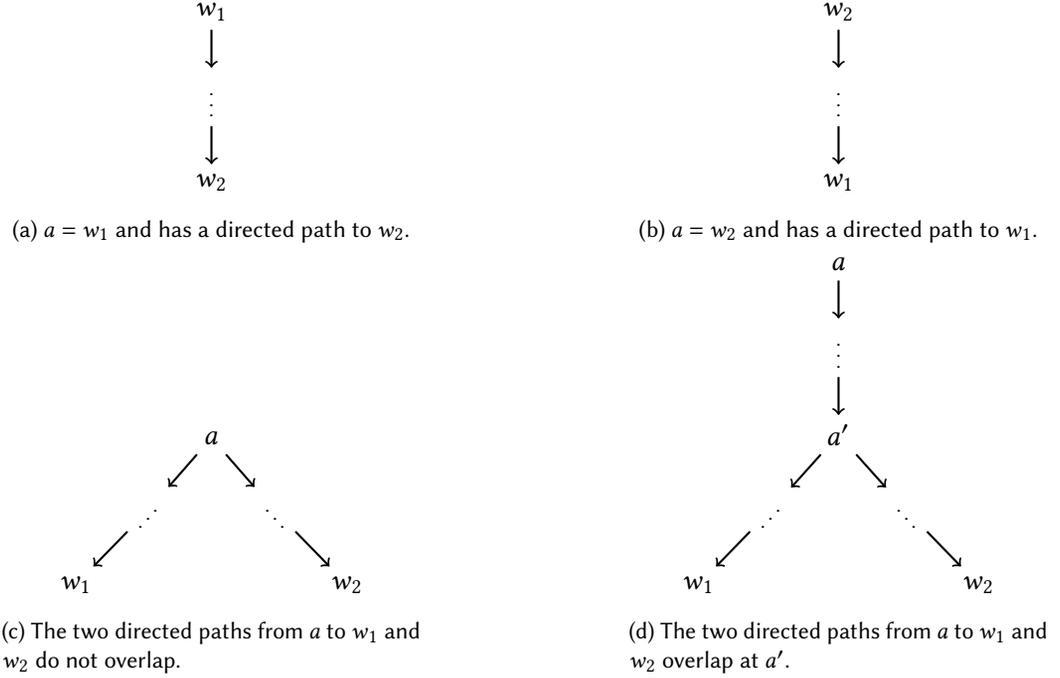

We now show that if two nodes share an unblocked ancestor, then they are $d$-connected.
\begin{lemma}\label{lem:confounder_path}
    Suppose distinct nodes $w_1, w_2$ share an unblocked ancestor $a$, i.e. $a \in \anc{w_1} \cap \anc{w_2}$. Then one of the following is true:
    \begin{enumerate}
        \item There is a $d$-connected, acyclic directed path $[w_1] {++} l {++} [w_2]$ such that $w_1 \not\in Z$.
        \item There is a $d$-connected, acyclic directed path $[w_2] {++} l {++} [w_1]$ such that $w_2 \not\in Z$.
        \item There is a $d$-connected, acyclic path $[w_1] {++} \mathsf{rev}(l_1) {++} [a'] {++} l_2{++} [w_2]$ such that both the paths $[a'] {++} l_1 {++} [w_1]$ and $[a'] {++} l_2 {++} [w_2]$ are directed.
    \end{enumerate}
\end{lemma}
\begin{proof}
    Since $a$ is a shared unblocked ancestor of $w_1, w_2$, we use \autoref{def:unblockedancestor} to produce the following cases:
    \begin{enumerate}
        \item $a = w_1$, shown in \autoref{fig:w1w2directed}. Then, since $w_1,w_2$ are distinct, $a\neq w_2$, so there is a directed path to $w_2$ that does not go through any member of $Z$. Then, the path is $d$-connected, since all intermediate nodes are mediators. Furthermore, since $a \in \anc{w_2}$ and $a \neq w_2$, $a \not\in Z$, as desired. Thus, Condition 1 is satisfied.
        \item $a \neq w_1$, $a = w_2$, shown in \autoref{fig:w2w1directed}. By symmetric logic as the case above, Condition 2 is satisfied.
        \item $a \neq w_1, a \neq w_2$. Then, there are directed paths from $a$ to $w_1$ and from $a$ to $w_2$, both of which do not go through members of $Z$. It is clear that $a \not\in Z$. If the two directed paths do not overlap, shown in \autoref{fig:confounderpatha}, then we simply concatenate them and satisfy Condition 3. If they do overlap, shown in \autoref{fig:confounderpathb}, then we find the first intersection $a'$ of the reverse paths given by \autoref{thm:first_elt_in_common}, and we take the path from $w_1$ to $a'$ via the reverse of $l_1$, then to $w_2$ via $l_2$. By \autoref{thm:first_elt_in_common}, this path is acyclic. Since the original directed paths do not pass through $Z$, this path is $d$-connected. Furthermore, it is clear that the two paths making up the resulting path are directed, since they are subpaths of the original directed paths.
    \end{enumerate}
    Thus, one of the conditions is satisfied in all cases.
\end{proof}

\section{Final Construction of the \texorpdfstring{$d$}{d}-Connected Path}\label{app:cons_path}

To make the reasoning of \autoref{sec:arb_length} precise, we now define functions to identify the set of $z \in Z$ whose values are affected during the sequence. These nodes serve as intermediaries in the propagation chain and will be used to construct an explicit $d$-connected path (specifically, its colliders) from $u$ to $v$.
\begin{minted}{coq}
Fixpoint find_unblocked_ancestors_in_Z_contributors {X: Type} `{EqType X}
        (G: graph) (Z: nodes) (AZ: assignments X) (S: nodes): nodes :=
  match AZ with
  | [] => []
  | (z, x) :: AZ' => if overlap (find_unblocked_ancestors G z Z) S
            then z :: find_unblocked_ancestors_in_Z_contributors G Z AZ' S
            else find_unblocked_ancestors_in_Z_contributors G Z AZ' S
  end.

Fixpoint get_conditioned_nodes_that_change_in_seq {X: Type} `{EqType X}
            (L: list (assignments X)) (Z: nodes) (AZ: assignments X) 
            (G: graph): nodes :=
  match L with
  | U1 :: L' => match L' with
                | U2 :: U3 :: L''' =>
    find_unblocked_ancestors_in_Z_contributors G Z AZ 
        (unblocked_ancestors_that_changed_A_to_B (nodes_in_graph G) U1 U2)
    ++ get_conditioned_nodes_that_change_in_seq L' Z AZ G
                | _ => []
                end
  | _ => []
  end.
\end{minted}
In the above, $\verb|unblocked_ancestors_that_changed_A_to_B|(\VV, U_1, U_2)$ returns every node with an unobserved term differing between $U_1$ and $U_2$. 
Let $$\Delta_Z(L) := \verb|get_conditioned_nodes_that_change_in_seq|(L, Z, A_Z, \GG).$$ In words, if we let $L = [U_0, U_1, ..., U_\ell]$, then $\Delta_Z(L)$ outputs the subset of $Z$ whose values were affected by changes in the sequence $U_2, \dots, U_\ell$. For simplicity, we will often instead refer to the same expression as $\Delta_Z(U_0, ..., U_\ell)$.

We now prove a result similar to \autoref{lem:value_affected_by_unblocked} tailored to sequences:
\begin{lemma}\label{lem:value_affected_by_unblocked_seq}
    For any graph function $f$ and sequence of unobserved-terms assignments $U_0, ..., U_\ell$ satisfying the conditions of \autoref{def:ci}, if $f_{U_0}(v) \neq f_{U_\ell}(v)$, then there exists a node $a\in \anc{v}$ such that one of the following is true:
    \begin{enumerate}
        \item $a\in \anc{u}$.
        \item There exists $z\in Z$ such that $a \in \anc{z}$ and $z\in \Delta_Z(U_0, ..., U_\ell)$.
    \end{enumerate}
\end{lemma}
\begin{proof}
    By \autoref{lem:value_affected_by_unblocked}, we know that there is a node $a \in \anc{v}$ such that $U_0(a) \neq U_\ell(a)$. If $a \in \anc{u}$, then Condition 1 is satisfied. Suppose $a \not\in \anc{u}$.

    Then, $U_0(a) = U_1(a)$, so $\ell \geq 2$. To show that Condition 2 above is satisfied for any $U_0, ..., U_\ell$ satisfying Condition 3 of \autoref{def:ci}, we perform induction on $\ell$. Note that for the rest of the proof, we no longer require that the sequence satisfy Condition 2 of \autoref{def:ci}.
    
    For the base case, $\ell = 2$. Then, $U_1(a) \neq U_2(a)$, so there must be a node $z\in Z$ and a node $a'$ such that $a \in \anc{z}$, $a' \in \anc{z}$, and $U_0(a') \neq U_1(a')$. Then, letting $$S := \verb|unblocked_ancestors_that_changed_A_to_B|(\VV, U_0, U_1),$$ we have that $$z \in \verb|find_unblocked_ancestors_in_Z_contributors|(\GG, Z, A_Z, S),$$ and thus $z \in \Delta_Z(U_0, U_1, U_2)$.

    For the induction hypothesis, we assume that the statement is true for $\ell = n$. Suppose the sequence has length $\ell = n + 1$, where $n \geq 2$. If $U_1(a) \neq U_2(a)$, proceed as in the base case. Otherwise, apply the induction hypothesis to the sequence $U_0' = U_1, U_1' = U_2, ...$, which tells us that there must be a node $z \in Z$ such that $a \in \anc{z}$ and $z \in \Delta_Z(U_1, U_2, ..., U_{n+1})$. Then, by definition of $\Delta_Z$, we have that $z \in \Delta_Z(U_0, U_1, ..., U_{n+1})$, as desired.
\end{proof}
This result allows us to pinpoint the change in $v$'s value to a specific conditioned node $z$. We now relate syntactic structure in the sequence to the nodes in $\Delta_Z(L)$.
\begin{lemma}\label{lem:z_correspond_to_U_sublist}
    For a node $z\in Z$ and $L = [U_0, ..., U_\ell]$, $z \in \Delta_Z(L)$ if and only if there exists a subsequence $U_i, U_{i+1}, U_{i+2}$ of $L$ such that $$z \in \verb|find_unblocked_ancestors_in_Z_contributors|(\GG, Z, A_Z, S),$$ where $S := \verb|unblocked_ancestors_that_changed_A_to_B|(\VV, U_i, U_{i+1})$. In words, $z$ was affected by changes between $U_i$ and $U_{i+1}$.
\end{lemma}
\begin{proof}
    In the mechanized proof, we would proceed via induction on $\ell$. However, it is clear from the Rocq function definitions that the statement is true.
\end{proof}
This lemma isolates the part of the sequence that causes a particular $z$ to change. It is possible for this change to come directly from $U_1$; if not, it is caused by reparative changes initiated by a different $z' \in Z$ in a previous transition, leading us to a recursive structure:
\begin{lemma}\label{lem:two_conditioned_nodes}
    Let $L = [U_0, ..., U_\ell]$ satisfy Condition 2 of \autoref{def:ci}. Suppose $U_{i}, U_{i+1}, U_{i+2}$ is the subsequence corresponding to $z\in \Delta_Z(L)$ as given by \autoref{lem:z_correspond_to_U_sublist}. Furthermore, suppose there exists $U'$ such that $U', U_i, U_{i+1}$ is a subsequence of $L$. Then, there exists node $a$ and node $z' \in Z$ such that $$a \in \anc{z} \cap \anc{z'}$$ and $$z' \in \verb|find_unblocked_ancestors_in_Z_contributors|(\GG, Z, A_Z, S'),$$ where $S' := \verb|unblocked_ancestors_that_changed_A_to_B|(\VV, U_{i-1}, U_i)$.
\end{lemma}
\begin{proof}
    We perform induction on $L$. For the base case, suppose $U' = U_0$, $U_i = U_1$, and $U_{i+1} = U_2$. Then, since $z$ is affected by changes from $U_1$ to $U_2$, there must be a node $a \in \anc{z}$ such that $U_1(a) \neq U_2(a)$. Then, there must be a node $z' \in Z$ such that $a \in \anc{z'}$, and $z'$ is affected by changes from $U_0$ to $U_1$, as desired.

    For the induction step, we assume that the result holds for $L' = [U_1, ..., U_\ell]$. If $U' = U_0$, $U_i = U_1$, and $U_{i+1} = U_2$ again, we follow the same steps as above. If not, the result follows directly from the induction hypothesis.
\end{proof}

\section{Concatenating Paths from Unobserved-Terms Assignments} \label{app:backward_overlap}

Recall that near the end of \autoref{sec:backward}, our goal is to construct a $d$-connected path from $u$ to $v$ using the existence of a sequence of unobserved-terms assignments $U_0, ..., U_\ell$ satisfying the conditions of \autoref{def:ci}, under the assumption that $f_{U_0}(v) \neq f_{U_\ell}(v)$.

The case where $\anc{u} \cap \anc{v} \neq \emptyset$ yields a $d$-connected path via \autoref{lem:confounder_path}. In the case where $\anc{u} \cap \anc{v} = \emptyset$, \autoref{lem:value_affected_by_unblocked_seq} guarantees the existence of some $z \in Z$ and $a \in \anc{z} \cap \anc{v}$ such that $z$ is affected by changes in the sequence. If $z$ is influenced directly by $U_1$, we can construct a $d$-connected path from $u$ to $z$. Otherwise, \autoref{lem:two_conditioned_nodes} provides a recursive structure in which the change to $z$ propagates through another node $z' \in Z$. Repeated applications of Lemmas~\ref{lem:z_correspond_to_U_sublist} and~\ref{lem:two_conditioned_nodes}, together with induction, allow us to construct a $d$-connected path from $u$ to $z$. Finally, we concatenate this path with a $d$-connected path from $z$ to $v$.

We now formalize the construction of a $d$-connected path from $u$ to such a node $z \in Z$.
Note that we will eventually have to concatenate this path with a $d$-connected path from $z$ to $v$, in which $z$ will have to be a collider to guarantee $d$-connectedness, as shown in \autoref{thm:concat_d_connected}.
Thus, we require the additional constraint that the final edge is into $z$.

\begin{lemma}\label{lem:path_u_to_z}
    Let $z\in Z$, and let $z\in \Delta_Z(U_0, ..., U_\ell)$. Then, there is a path $P_{u,z}$ from $u$ to $z$ that is acyclic, $d$-connected given $Z$, and into $z$ at the last edge.
\end{lemma}
\begin{proof}
    By \autoref{lem:z_correspond_to_U_sublist}, there must be a subsequence $U', U'', U'''$ of $U_0, U_1, ..., U_\ell$ such that $z$ was affected by changes between $U'$ and $U''$. Let $i$ be the index of the subsequence in the sequence (the $i$-th entry of the sequence is $U'$, the $(i+1)$-th is $U''$, and the $(i+2)$-th is $U'''$). Note further that this index does not have to be unique; we do not require that the sequence of unobserved-terms assignments are minimal. We can choose any index $i$ that satisfies the requirements.

    Now, we proceed via strong induction on $i$. For the base case, suppose that $i = 0$, so $U' = U_0, U'' = U_1, U''' = U_2$. Then, there must exist some $a$ such that $a \in \anc{z}$, and $U_0(a) \neq U_1(a)$. Thus, $a \in \anc{u}$. Since $z \in Z$, one of Conditions 1 or 3 of \autoref{lem:confounder_path} must be true. Either way, it is clear that the path goes into $z$, is $d$-connected, and is acyclic.

    For the induction hypothesis, suppose that for any $z' \in Z$ that is affected by changes between assignments appearing at index at most some $n$, where $n \geq 0$, there is a path $P_{u,z'}$ from $u$ to $z'$ that is acyclic, $d$-connected given $Z$, and into $z'$. Assume $i = n + 1 > 0$. Then, there must be some $U$ that immediately precedes the subsequence $U', U'', U'''$; in other words, $U$ appears at index $(i-1)$ in the sequence. Then by \autoref{lem:two_conditioned_nodes}, there is a $z' \in Z$ and a node $a$ such that $a \in \anc{z} \cap \anc{z'}$, and $z$ is affected by changes between $U$ and $U'$. It is clear that $(i-1)$ is a valid index of the subsequence $U, U', U''$ in the sequence. Thus, we apply the induction hypothesis to get a path $P_{u,z'}$ that goes into $z'$.

    Furthermore, since $z,z' \in Z$, Condition 3 of \autoref{lem:confounder_path} must be true, giving us some single-confounder path $P_{z',z}$ from $z'$ to $z$. We aim to concatenate these two paths; however, we must consider the possibility that they intersect at some point. Luckily, the casework here is much easier than in \autoref{app:disjoint_desc_paths}.

    If $P_{u,z'}$ and $P_{z',z}$ do not intersect, we concatenate them directly. Since the last edge of $P_{u,z'}$ is into $z'$, $z'$ is a collider, and the path remains $d$-connected because $z' \in Z$.

    If $u$ lies on $P_{z',z}$, then we simply take the subpath of $P_{z',z}$ starting from $u$ and continuing to $z$. Since $P_{z',z}$ is acyclic, $d$-connected, and into $z$, so is this subpath. Similarly, if $z$ lies on $P_{u,z'}$, then we simply take the subpath of $P_{u,z'}$ starting from $u$ and continuing to $z$; this subpath must be acyclic and $d$-connected. Since $z\in Z$, and $P_{u,z'}$ is $d$-connected, $z$ must be a collider in the path. Thus, the subpath is also into $z$.

    Otherwise, let $x$ be the first intersection point from \autoref{thm:first_elt_in_common} of $P_{u,z'}$ and the reverse of $P_{z',z}$. Construct the new path by taking $P_{u,z'}$ to $x$, then continuing along $P_{z',z}$. Let this composite path be $P_{u,z}$. This path is acyclic by \autoref{thm:first_elt_in_common} and still into $z$. Note that since $x\in P_{z', z}$, it must be a mediator or confounder in $P_{z',z}$ and thus not in $Z$. In $P_{u,z}$, if $x$ is a mediator or confounder, then the path is $d$-connected, since $x\not\in Z$. If $x$ is a collider in $P_{u,z}$, then we consider where $x$ falls in $P_{z',z}$. In particular, if $x$ falls before the confounder, then $x$ has a descendant path to $z'$. Otherwise, $x$ has a descendant path to $z$. Either way, $P_{u, z}$ is $d$-connected.

    All the above cases result in paths $P_{u,z}$ satisfying the requirements in the lemma statement.
\end{proof}

We now have all the pieces to prove the backward direction \autoref{thm:equiv} (restated below for clarity), thus completing the full proof of the theorem.

\equiv*
\begin{proof}[Proof (Backward).]
    We aim to show that if $u$ and $v$ are $d$-separated given $Z$, then they are semantically separated given $Z$.

    We prove the contrapositive. Suppose that $u$ and $v$ are not semantically separated, so there exists some graph function $f$, some $\alpha \neq \beta$, and a sequence of unobserved-terms assignments $U_0, ..., U_\ell$, where $\ell \geq 1$, satisfying the conditions of \autoref{def:ci}, such that $f_{U_0}(v) \neq f_{U_\ell}(v)$. We will show that $u$ and $v$ are not $d$-separated given $Z$. 

    By \autoref{lem:value_affected_by_unblocked_seq}, there exists an $a \in \anc{v}$ such that either $a \in \anc{u}$, or there exists $z\in Z$ such that $a \in \anc{z}$ and $z \in \Delta_Z(U_0, ..., U_\ell)$. Consider the first case. Then, by \autoref{lem:confounder_path}, there is a $d$-connected path from $u$ to $v$, and thus $u$ and $v$ are not $d$-separated.

    In the second case, we know by \autoref{lem:path_u_to_z} that there is a $d$-connected, acyclic path $P_{u, z}$ from $u$ to $z$ that goes into $z$ at the last edge. Furthermore, since $z$ and $v$ share unblocked ancestor $a$, and $z\in Z$, one of Conditions 2 of 3 of \autoref{lem:confounder_path} must hold, and thus there is a $d$-connected path $P_{z,v}$ from $z$ to $v$ that goes into $z$ at the first edge. 
    
    We now consider the structure of the path from $z$ to $v$ and how it intersects with the previously constructed path from $u$ to $z$.
    \begin{enumerate}
    \item If $P_{u,z}$ and $P_{z,v}$ do not intersect anywhere, then we can simply concatenate them to get $P_{u,v}$. Since $P_{u,z}$ is into $z$, $z$ becomes a collider in $P_{u,v}$, so $P_{u,v}$ is $d$-connected.
    \item If $u$ lies on $P_{z,v}$ or $v$ lies on $P_{u,z}$, we take the appropriate subpath, which remains a $d$-connected path.
    \item If the first overlap of $P_{u,z}$ and the reverse of $P_{z,v}$ given by \autoref{thm:first_elt_in_common} is some node $x \not\in \{ u,v \}$, then we let $P_{u,v}$ be the path $P_{u,z}$ until $x$, then switch to the path $P_{z,v}$ until $v$. This path is acyclic by \autoref{thm:first_elt_in_common}. Since $x \in P_{z,v}$, it is a mediator or confounder and thus not in $Z$. Thus if $x$ is a mediator or confounder in $P_{u,v}$, then $P_{u,v}$ is still $d$-connected. Suppose $x$ is a collider in $P_{u,v}$.

    Consider the structure of $P_{z,v}$ as given by \autoref{lem:confounder_path} (it must satisfy Condition 2 or Condition 3). If $P_{z,v}$ is a directed path from $v$ to $z$, then the descendant path from $x$ to $z$ along the reverse of $P_{z,v}$ ensures that $P_{u,v}$ is $d$-connected.

    If $P_{z,v}$ is a single-confounder path from $z$ to $v$, then $x$ must come before the confounder in $P_{z,v}$, since otherwise, the edge on the right of $x$ in $P_{u,v}$ would be outwards, and thus $x$ would not be a collider. Thus, the descendant path from $x$ to $z$ along the reverse of $P_{z,v}$ still ensures that $P_{u,v}$ is $d$-connected.
    \end{enumerate}
    Thus, in all possible overlapping cases, there is a path $P_{u,v}$ from $u$ to $v$ that is $d$-connected given $Z$, so $u$ and $v$ are not $d$-separated, as desired.
\end{proof}
\end{document}